\documentclass[lettersize,journal]{IEEEtran}
\usepackage{amsmath,amsfonts}
\usepackage{algorithmic}
\usepackage{algorithm,algorithmic}
\usepackage{array}
\usepackage[caption=false,font=normalsize,labelfont=sf,textfont=sf]{subfig}
\usepackage{textcomp}
\usepackage{stfloats}
\usepackage{url}
\usepackage{verbatim}
\usepackage{graphicx}
\usepackage{cite}
\usepackage{color}
\usepackage{amsthm}
\usepackage{textcomp}
\newtheorem{assumption}{Assumption}
\newtheorem{remark}{Remark}
\newtheorem{proposition}{Proposition}
\newtheorem{lemma}{Lemma}
\newtheorem{theorem}{Theorem}
\begin{document}

\title{Consensus Tracking Control of Multi-agent Systems with A Time-varying Reference State under  Binary-valued Communication}
       
\author{Ting Wang,~\IEEEmembership{Member,~IEEE,} Zhuangzhuang Qiu, Xiaodong Lu,
Yanlong Zhao,~\IEEEmembership{Senior Member,~IEEE}

\thanks{This paper  was supported by National Natural Science Foundation of China under Grants 62025306 and 62473040.}
\thanks{Ting Wang and Zhuangzhuang Qiu are with the School of Intelligence Science and Technology, University of Science and Technology Beijing, and also
	with the Institute of Artificial Intelligence, University of Science and Technology Beijing, Beijing, 100083,
	China (e-mails: wangting@ustb.edu.cn, m15201340427@163.com).}
\thanks{Xiaodong Lu is with the School of Automation and Electrical Engineering, University of Science and Technology Beijing, Beijing, 100083,
	China (e-mails:luxiaodong@amss.ac.cn).}
\thanks{Yanlong Zhao is with the LSC, NCMIS, Academy of Mathematics, and Systems
Science, Chinese Academy of Sciences, Beijing, 100190, China, and also
with the School of Mathematical Sciences, University of Chinese Academy
of Sciences, Beijing, 100049, China (e-mail: ylzhao@amss.ac.cn).}
}



\maketitle

\begin{abstract}
This paper investigates the problem of consensus tracking control of discrete-time multi-agent systems under binary-valued communication. Different from most existing studies on consensus tracking, the transmitted information between agents is the binary-valued. Parameter identification with binary-valued observations is applied to the estimation of neighbors' states and the tracking control is designed based on the estimation. 
Two Lyapunov functions  are constructed to deal with the strong coupling of estimation and control. 
Compared with consensus problems under binary-valued communication, a reference state is required for consensus tracking control. 
Two scenarios  of the time-varying reference state are studied respectively.  (1) The reference state is asymptotically convergent. An online algorithm that performs estimation and control simultaneously is proposed, in which the estimation step size and the control gain are decreasing with time. By this algorithm, the multi-agent system is proved to achieve consensus tracking with convergence rate $O(1/k^\varepsilon)$ under certain conditions. (2) The reference sate is  bounded, which is less conservative than that in the first case. In this case, the esimation step size and control gain are designed to be constant.  By this algorithm,  all the followers can reach to a  neighborhood of the leader with an exponential rate.  Finally, simulations are given to demonstrate theoretical results. 


\end{abstract}


\section{Introduction}
\IEEEPARstart{T}{he} multi-agent system (MAS) is composed of interacting individuals or agents, each of which is a physical or virtual entity. MAS can complete some complex tasks by communication, cooperation or compete among the agents.  It is widely used in real life including cooperative control of unmanned air vehicles, robot arm coordination, cooperative sensor measurement and smart micro-grids {\cite{ref1,ref2,ref3,ref4}}. 

One task of multi-agent systems is to design an appropriate controller to achieve consensus tracking. In \cite{ref5, ref6}, some basic methods are proposed to solve the consensus tracking control problem of multi-agent systems. The sufficient and necessary conditions of convergence have been widely used in the hypothesis of network topology mechanisms. In \cite{ref7, ref8, ref9}, a cooperative control framework and tracking control algorithm for general linear systems with a time-invariant leader are presented. It is proved that the consensus tracking can be achieved by containing spanning trees in a directed graph. Refs. \cite{ref10, ref11} studied the consensus tracking problem of linear multi-agent systems with  switching  directed topology and communication delay. The linear matrix inequality and Lyapunov stability theory are used to prove that the tracking algorithm can achieve exponential convergence and ensure the boundedness of tracking errors. 

The above articles considered that each agent can obtain the precise states of neighbors through the signal flow in the topology. However, the information available to us is often imprecise or quantized due to the limitation of communication bandwidth. Many scholars have proposed solutions to consensus  tracking problems of MAS with quantized communication.  Integer quantization is considered in\cite{ref12}, Akshay Kashyap et al. proposed consensus tracking algorithms and gave bounds on the consensus tracking time of these algorithms with fully linear connected networks. For infinite level quantization, Ruggero Carli et al.\cite{ref13} analyzed the feasibility of the  consensus tracking algorithm and estimated the probability of the error according to the eigenvalues of the evolution matrix describing the algorithm. The study of finite level quantization started with Refs \cite{ref14, ref15, ref16}. In \cite{ref14}, Francesca Ceragioli et al. studied asymptotic and finite-time convergence conditions for first-order MAS. In \cite{ref15}, Zhihong Guan et al. derived some conditions for quantized consistency of second-order multi-agent systems based on stability theory. In \cite{ref16}, Weisheng Chen extended the conclusion of convergence to the necessary and sufficient conditions for sampling period and design parameters. Junjie Fu\cite{ref17} gave the comprehensive influence of three-level quantized correlation information error and quantized absolute information error on tracking results.  

As far as we know,  consensus tracking  of multi-agent systems based on binary-valued information is rarely studied, although there are some results on consensus control based on binary-valued information \cite{ref18,refGeorge,Ke2023}. Tracking a constant reference state  (state of leader) under binary-valued communication can be transformed as a consensus problem under binary-valued communication \cite{refCAC}. However, tracking a time-varying objective is more difficult than that with a constant objective, especially by using limited information.  

This paper studies consensus tracking problem of multi-agent systems with a time-varying reference state under binary-valued communication. Two kinds of reference states are considered. One is asymptotically convergent, the other is bounded. 
The contributions of this paper are given as follows.
\begin{enumerate}
\item Consensus tracking of multi-agent systems with a time-varying reference state is first studied  under binary-valued communication. Different from consensus tracking under general quantized communication \cite{ref14, ref15, ref16, ref17}, binary-valued communication needs lower bandwidth  and provides less information. It is more difficult to achieve consensus tracking by using less information. An identification method with binary-valued observations is applied to the estimation of agents' states and the tracking control is designed based on the estimation. Since the estimate error and tracking error are strongly coupled, two Lyapunov functions are constructed to give the convergence and convergence rate of the consensus tracking algorithms.
 \item Different change rates of the time-varying reference state lead to different consensus tracking errors between leader and followers, which affect estimation errors and bring difficulties in analysis of algorithms. For the case of the convergent reference state, an algorithm with a decaying step size of estimate and a decaying control gain are proposed. By this algorithm, it is proved that all the followers' states can converge to the reference state under certain conditions. And, the convergence rate can be given by the change rate of reference state. The slower the change rate of reference state, the faster the convergence rate of the algorithm. 
Moreover, the algorithm can degenerate into that in \cite{refCAC} for a constant reference state and the corresponding results are consistent with those in \cite{refCAC}.
\item  For the case of the  bounded reference state, decaying gains cannot ensure the convergence of tracking error. The step sizes and control gains are designed to be constant. By choosing appropriate values of step sizes and control gains, the estimate error and tracking error can be compressed. It is proved that all the followers' states can reach to a  neighborhood of the reference state with an exponential rate and the range of the neighborhood depends on the change rate  of reference state.
	
\end{enumerate}

The rest of this paper is organized as follows. Section \ref{sec:prfor} formulates the consensus tracking problem of multi-agent systems under binary-valued communication. Section \ref{sec:cov} considers  case of asymptotically convergent reference state. An algorithm and its properties are given.   Section \ref{sec:bod} considers the case of  bounded reference state. In Section \ref{sec:sim}, numerical simulations are given to illustrate the effectiveness of  theoretical algorithms. Section \ref{sec:con} is the conclusion of this paper.

\section{Problem Formulation}\label{sec:prfor}
Consider a discrete time MAS consisting of $n$ agents 
\begin{align}
	\label{eq:agent}
	x_i(k+1)=x_i(k)+u_i(k), \; i=1,\ldots, n,
\end{align}
where $x_i(k)\in R$ and $u_i(k)\in R$ are respectively the state and the control input
of $i$th agent at time $k$. The $n$ agents are called followers. There is a leader  labeled agent $n+1$. It has the information state $x_{n+1}(k)$, which is called the reference state. The reference state is time-varying, which satisfies
\begin{align}
	\label{eq:leader}
    x_{n+1}(k+1)=x_{n+1}(k)+f(k),	
\end{align}
where $f(k)$ is the change rate of the leader's state. 

The MAS consisting of a leader and $n$ followers is studied. 
Combining equation \eqref{eq:agent} and  \eqref{eq:leader}, the system can be written as  
\begin{align}
	x_i(k+1)=x_i(k)+u_i(k),\; i=1,\ldots, n+1,
\end{align}
where $u_{n+1}(k)=f(k)$, $k=1,2,\ldots$. Its vector form is as follows:
\begin{align}
	x(k+1)=x(k)+U(k),
\end{align}
with $x(k)=\left[x_1(k), \ldots, x_{n+1}(k)\right]^T$ and $U(k)=\left[u_1(k), \ldots,u_{n+1}(k)\right]^T$.

Multi-agent systems are distributed by network topology $\mathcal{G} = (\mathcal{N} , \mathcal{E} )$, where $\mathcal{N} = \{1, 2, . . . , n+1\}$ is the node set, i.e. the set of agents, and $\mathcal{E} = N\times N$ is the edge set, i.e. the set of signal transmissions. Matrix $ \mathcal{A}=(a_{ij})_{(n+1)\times (n+1)}$ is the Adjacency Matrix of topology $\mathcal{G}$. The element $a_{ij}=1$ if there is an edge from $j$ to $i$, else $a_{ij}=0$. If $a_{ij}=1$, then agent $j$ is called the neighbor of agent $i$, denoted as $j\in\mathcal{N}_i$, which means agent $i$ can get the information from agent $j$. Since the leader cannot receive feedback from any followers,  $a_{(n+1)j}=0, \; j=1,2,...,n$. Denote $d_i$ as the number of agent $i$’s neighbors, $\mathcal{D}$ = diag$(d_1, d_2, . . . , d_{n+1})$. 
The Laplacian of directed graph is defined as $\mathcal{L} = \mathcal{D} - \mathcal{A}$.

The signal flow from agent $j$ to agent $i$ in the topology is affected by random noise. Each agent $i$ $(i=1, ..., n+1)$ can only get the binary-valued information from its neighbors $j$ as follows.
\begin{align}
	\begin{cases}
		y_{ij}(k)=x_j(k)+\xi_{ij}(k), \quad j=1,2,\ldots, n+1,\\
		\mathcal{S}_{ij}(k)=I_{\lbrace y_{ij}(k)\leq \mathcal{B}_{ij}\rbrace }, \quad j\in \mathcal{N}_i,
	\end{cases}
	\label{eq:observation}	  
\end{align}
where $x_j(k)$ is the state of agent $j$, $\xi_{ij}(k)$ is the external noise, subject to the assumed distribution, $y_{ij}(k)$ is the immeasurable output, $\mathcal{B}_{ij}$ is the threshold of a binary-valued sensor, $\mathcal{S}_{ij}(k)$ is the binary-valued information that agent $i$ obtains from agent $j$,  $I_{\lbrace . \rbrace }$ is an indicative function, defined as 

\begin{align} 
	I_{\lbrace v \in V\rbrace }={\begin{cases}1,\quad {\text{if}}\; v \in V, \\ 0, \quad {\text{others}}. \end{cases}}
\end{align}
The objective of this work is to design distributed controllers $u_i(k), i=1, \ldots, n$ such that the followers' states can track the leader's state under binary-valued communication in the following sense.
\begin{align}
	\lim\limits_{k\to\infty}\left\Vert x_i(k)-x_{n+1}(k)\right\Vert= 0,\quad i=1,...,n, 
\end{align} 
where $\|x\|=E(x^Tx)$, for any random  variable $x$.

The topology and noise of the system need to satisfy the following assumptions.

\begin{assumption}\label{ass1}
	The topology contains a directed spanning tree with the leader as the root node.
\end{assumption}

\begin{remark}
	Let $\mathcal{L}$ be the Laplacian matrix of a directed graph $\mathcal{G}$, we can get that $\mathcal{L}$ has the following properties. 
	\begin{enumerate}
		\item The eigenvalues of the matrix $\mathcal{L}$  all have nonnegative real parts.
		\item The matrix $\mathcal{L}$ has a unique eigenvalue of $0$, and its corresponding eigenvector is $1_{n+1}$. Here $1_{n+1}$ is a $n+1$-dimensional vector whose entries are all $1$.
	\end{enumerate}
\end{remark}
\begin{remark}\label{rem2}
	If the network topology contains a directed spanning tree, there exists a matrix $\psi_{(n+1)\times n} $ that satisfies the following conditions,
	\begin{equation*}
		rank\{\psi_{(n+1)\times n}\} =rank\{\mathcal{L}\},
	\end{equation*}
	then we can define an orthogonal matrix $\Theta=[q,\psi_{(n+1)\times n }]$, $\Theta^{-1}=\begin{pmatrix}
		\tau \\ 
		\phi_{n\times (n+1)}
	\end{pmatrix}$, and $\Theta^{-1}\mathcal{L}\Theta=\begin{pmatrix}
		0& 0_{1\times n}\\
		0_{n\times 1}& \tilde{\mathcal{L}}
	\end{pmatrix}$, where the matrix $\tilde{\mathcal{L}}$ is a Joldan matrix, and all the eigenvalues of $\tilde{\mathcal{L}}$ are non-zero eigenvalues of $\mathcal{L}$ with positive real parts.
\end{remark}

\begin{assumption}\label{ass2}
	The noises $\{\xi_{ij} (k), i \in N_j\}$  are independent of the state of the system and follows a normal distribution with mean $0$.
	The corresponding distribution function and probability density function are respectively $\mathcal{F}(\cdot)$ and $f(\cdot)$, where $f(x) = d\mathcal{F}(x)/dx \neq  0$.
\end{assumption}

\begin{remark}
	We assume that the noise distribution in Assumption \ref{ass2} is known; if not, the noise distribution can be estimated according to the method in reference \cite{ref19} .
\end{remark}



\section{The reference state is asymptotically convergent}\label{sec:cov}
In this section, we assume the sate of leader is asymptotically convergent, which can be described as the following assumption. 
\begin{assumption}\label{ass:cov}
There exists a constant $x^*$ such that
$$\lim_{k\rightarrow\infty}x_{n+1}(k)=x^*.$$
\end{assumption}

By Assumption \ref{ass:cov}, we can get the following lemma.
\begin{lemma}\label{lem:f}
	The  change rate of reference state satisfies that $\sum_{i=1  }^{\infty} f(i)$ is asymptotically convergent.
\end{lemma}
\begin{proof}
	By \eqref{eq:leader}, we have
	\begin{align*}
	\sum_{i=1 }^{k} f(i)
	= & f(k)+f(k-1)+\ldots+f(1)   \\
	=&x_{n+1}(k+1)-x_{n+1}(k)+x_{n+1}(k)-x_{n+1}(k-1)\\
	&+...+x_{n+1}(2)-x_{n+1}(1)  \\
	=& x_{n+1}(k+1)-x_{n+1}(1).  
	\end{align*}
	By Assumption \ref{ass:cov}, we have
	$$\lim_{k\rightarrow\infty}\sum_{i=1 }^{k} f(i)=\lim_{k\rightarrow\infty}x_{n+1}(k+1)-x_{n+1}(1)=x^*-x_{n+1}(1),$$
	which means $\sum_{i=1}^{\infty} f(i)$ is asymptotically convergent.
	
\end{proof}


In this section, the recursive projection algorithm (RPA) proposed in \cite{ref20} is used to estimate the states of the neighbors, where the step size is decreasing. In addition, the controller is designed with a decaying gain to reduce the impact of noise. The consensus tracking algorithm of a convergent reference state (CRS) are given as the following five steps:

	\begin{enumerate}
	\item Start: Initializes the state of each agent $i$ and estimates of its neighbors in the following form:
	\begin{equation*} 
	x_i(0)=x_i^0,\quad \hat{x}_{ij}(0)=\hat{x}_{ij}^0, 
	\end{equation*}
	where $i=1,2,\ldots, n, j\in \mathcal{N}_i$, $\left|x_i^0\right|\le \mathcal{W}$ and $\left|\hat{x}_{ij}^0\right|\le \mathcal{W}$.
	\item Observation: Each agent $i$ obtains the binary-valued observed states of its neighbors, as equation \eqref{eq:observation}.
	\item Estimation: Each agent $i$ uses the observed binary-valued measurements to estimate its neighbors' states based on the RPA :
	\begin{align} 
	\hat{x}_{ij}(k)=
	&\Pi _\mathcal{W}\Big\{ \hat{x}_{ij}(k-1 )\vphantom{\frac{\beta }{k}} \notag\\ 
	&+\frac{\beta }{k}\left(\mathcal{F}\left(\mathcal{B}_{ij}-\hat{x}_{ij}(k-1)\right)-\mathcal{S}_{ij}(k)\right)\Big\},
	\label{up:estimate}
	\end{align}
	where $\beta$ is a coefficient in step size of estimation, $\Pi _\mathcal{W}\left\lbrace .\right\rbrace$ is a
	projection operator, which is defined as follows:
	\begin{equation*}
	\Pi _\mathcal{W}(x)=\arg \underset{\left|m\right|\leq \mathcal{W}}{\min} \left| x-m \right|={\begin{cases}-\mathcal{W},\quad &{\text{if}}\; x< -\mathcal{W};\\ x,\quad &{\text{if}}\;|x|\leq \mathcal{W};\\ \mathcal{W},\quad &{\text{if}}\; x>\mathcal{W}. \end{cases}}
	\end{equation*}
	\item Control: Based on the estimates, each agent designs the
	control with a decreasing tracking gain $1/(k + 1)$:
	\begin{align*}
	u_i(k)=&\frac{-1}{k+1}\sum_{j=1  }^{n+1}a_{ij}\left[x_i(k)-\hat{x}_{ij}(k)\right],\; i=1,2, \ldots ,n.
	\end{align*}
	By designing the controller, the state of the agent $i$ is updated as follows
	\begin{align} 
	x_i(k+1)&=x_i(k)-\frac{1}{k+1}\sum_{j=1  }^{n+1}a_{ij}\left[x_i(k)-\hat{x}_{ij}(k)\right].
	\label{agentupdate1}
	\end{align}
	\item Let $k = k + 1$, go back to the first step.
\end{enumerate}

The algorithm can ensure the boundedness of followers' states and the estimates of neighbors' states, which are given as the following remarks. 
\begin{remark}\label{rem:estimate}
	Due to the definition of the projection operator,
	we can get that the estimates of neighbors's states are bounded.
	\begin{equation*}  	
		\left|\hat{x}_{ij}(k)\right| \leq \mathcal{W},\; \forall \; i=1,2,\ldots, n, \; j\in N_i.
	\end{equation*}
\end{remark}
\begin{proposition}[\cite{ref18}, Proposition 1]\label{rem:pro1}
	The states of all the followers satisfy
	\begin{equation*} 	
		\left|x_i(k)\right| \le \mathcal{W},\quad \forall\; i=1, \dots, n, \; k\geq  d_*,
	\end{equation*}
	where  $d_*=\max\left\{d_1,d_2,\dots,d_n\right\}$, $d_i$ is the number of agent $i$'s neighbors. 
\end{proposition}

Let $\theta_{ij}(k)=\hat{x}_{ij}(k)-x_j(k)$. Define the error vectors $\theta(k)$  in a given order as follow
\begin{align}\label{thetak}
	\theta (k)=&(\theta _{1r_1}(k), \ldots,\theta _{1r_{d_1}}(k),\theta _{2r_{d_1+1}}(k),\ldots,\theta _{2r_{d_1+d_2}}(k) \notag\\ &\ldots,\theta _{nr_{d_1+\cdots +d_{n-1}+1}}(k),\ldots,\theta _{nr_{d_1+\cdots +d_{n}}} (k))^T ,
\end{align}
where $r_1,\ldots, r_{d_1}\in \mathcal{N}_1,\quad  r_{d_1+1},\ldots, r_{d_1+d_2}\in \mathcal{N}_2,\cdots, r_{d_1+\cdots +d_{n-1}+1},\ldots, r_{d_1+\cdots +d_n}\in \mathcal{N}_{n}$. 
For any $j \in \mathcal{N}_i, i = 1, . . . , n$, we can define $n+1$ dimensional vectors $m_{ij}$ 
and $w_{ij}$ respectively as follows
\begin{align*} 
	m_{ij}= &(0,\ldots,0,\underbrace{1}_{i \rm {th\;    position}},0,\ldots,0)^T.
\end{align*}
$$w_{ij}= (0,\ldots,0,\underbrace{1}_{j \rm{th\;position}},0,\ldots,0)^T,$$ 
Putting $\{m_{ij} , j \in \mathcal{N}_i, i = 1, . . . , n\}$ 
and $\{w_{ij} , j \in \mathcal{N}_i, i = 1, . . . , n+1\}$ 
in the same order as $\theta(k)$ in
eq. \eqref{thetak}, we can get matrix $M $ 
and matrix $W $ are as follows:
\begin{align}
	M=&\Big[m_{1r_1}, \ldots,m_{1r_{d_1}},m_{2r_{d_1+1}},\ldots,m_{2r_{d_1+d_2}},\ldots, \\ &\quad m_{nr_{d_1+\cdots +d_{n-1}+1}},\ldots,m_{nr_{d_1+\cdots +d_n}} \Big] \notag\\
	 =&\begin{bmatrix}\underbrace{1\;  \cdots \; 1 }_{d_1}\\ &\underbrace{1 \;\cdots \;1 }_{d_2}& & \text{0} \\ \text{0} & & \cdots \\ & & & \underbrace{1\;  \cdots \; 1 }_{d_{n}} \\\text{0} & \text{0} & \cdots & 0 \end{bmatrix}_{(n+1)\times (d_1+\cdots +d_{n})},
	 \label{def:M}
\end{align}
\begin{align}\label{def:W}
W= &\begin{bmatrix}w_{1r_1}^T&\\ \vdots &\\ w_{r_{d_1}1}^T&\\ w_{r_{d_1+1}2}^T&\\ \vdots &\\ w_{r_{d_1+\cdots +d_{n+1}}(n+1)}^T&\\ \end{bmatrix}_{(d_1+\cdots +d_{n+1})\times (n+1)}.
\end{align}

Let $x(k)=[x_1(k), x_2(k), \ldots, x_{n+1}(k)]^T$. The vector form of system updating can be given  as follows by equations \eqref{eq:leader} and \eqref{agentupdate1}.
\begin{equation} 
	x(k+1)=\left(I-\frac{ \mathcal{L}}{k+1}\right)x(k)+\frac{ M}{k+1}\theta(k)+\Gamma(k),
	\label{equation}
\end{equation}
where $\mathcal{L}$ is the $(n+1)\times (n+1)$ Laplacian Matrix of the network $\mathcal{G}$, and $\Gamma(k)=\left[0, 0, \ldots, 0, f(k)\right]^T$.

\begin{remark}
	Comparing with the consensus algorithm in \cite{ref18}, the vector form of updating \eqref{equation} has one more item $\Gamma(k)=\left[0, 0, \ldots, 0, f(k)\right]^T$. So, the properties of $f(k)$  will affect the properties of the algorithm. 
\end{remark}


Let
\begin{align}
	\xi(k)&=\left(x_1-\sum_{j=1  }^{n+1}\pi_jx_j,x_2-\sum_{j=1  }^{n+1}\pi_jx_j,...,x_{n+1}-\sum_{j=1  }^{n+1}\pi_jx_j\right)\notag\\&=\left(I-J_{n+1}\right)x(k),
\end{align}
where $J_{n+1}=1_{n\times 1}\pi$, $\pi=(\pi_1, \pi_2, ..., \pi_{n+1})$ is defined in Remark  \ref{rem2}, $\pi$ is the left eigenvector of Laplacian matrix $\mathcal{L}$ with eigenvalue $0$, and $\pi1_{n+1}=1$.
Since $\left(I-J_{n+1}\right)\left(I-\frac{\mathcal{L}}{k+1}\right)=\left(I-\frac{\mathcal{L}}{k+1}\right)\left(I-J_{n+1}\right)$, we have by \eqref{equation}
\begin{align*}
	\xi(k+1)
	=&\left(I-\frac{\mathcal{L}}{k+1}\right)\xi(k)\\
	&+\frac{\left(I-J_{n+1}\right)M}{k+1}\theta(k)+\left(I-J_{n+1}\right)\Gamma(k).
\end{align*}
Let $\tilde{\xi}(k)=[\tilde{\xi}_1(k),\tilde{\xi}_2(k)...,\tilde{\xi}_{n+1}(k)]^T=\Theta^{-1}\xi(k)$, where  $\Theta^{-1}=\begin{pmatrix}
\tau \\ 
\phi_{n\times (n+1)}
\end{pmatrix}$ is defined in Remark \ref{rem2},  and $ \eta(k)=[\tilde{\xi}_2(k),\tilde{\xi}_3(k),...,\tilde{\xi}_{n+1}(k)]^T$. Then, we have  
$$\tilde{\xi}_1(k+1)=\pi\xi(k+1)=\pi\left(I-J_{n+1}\right)x(k)=0,$$
and 
	\begin{align}
		&\eta(k+1)\notag\\
		=&\left(I-\frac{\tilde{\mathcal{L}}}{k+1}\right)\eta(k)+\frac{\phi(I-J_{n+1})M}{k+1}\theta(k)\notag\\
		&+\phi\left(I-J_{n+1}\right)\Gamma(k),
	\label{update3}
	\end{align}
where $\phi=\phi_{n\times (n+1)}$ for simplicity. 
We can conclude that $\left\Vert\xi(k)\right\Vert=\left\Vert \tilde{\xi}(k)\right\Vert=\left\Vert \eta(k)\right\Vert$, where $\|x\|=\sqrt{E(x^Tx)}$ for any random variable $x$.
If $\lim\limits_{k\to\infty} \|\xi(k)\|=\lim\limits_{k\to\infty}\sqrt{\sum_{i=1  }^{n+1}\left\|\sum_{j=1  }^{n+1}\pi_j(x_i-x_j)\right\|^2}=0$, we can get that
$$\lim\limits_{k\to\infty}\left\| x_i-x_{n+1} \right\|=0,\quad \forall\; i=1, \ldots, n.$$
Hence,  the system can achieve consensus tracking if $\lim\limits_{k\to\infty}\left\Vert \eta(k)\right\Vert=0$.

Since $-\tilde{\mathcal{L}}$ is a Hurwitz matrix, given any positive number $\kappa$, there exists a  matrix $H>0$ such that 
	\begin{equation}\label{unique}
		H\tilde{\mathcal{L}}+\tilde{\mathcal{L}}^TH=\kappa I.
	\end{equation}
Duo to $\lambda_{min}(H)E\left(\eta(k)^T\eta(k)\right)\le E\left(\eta(k)^TH\eta(k)\right)\le \lambda_{max}(H)E\left(\eta(k)^T\eta(k)\right)$, we can get 
\begin{equation}
	\lim\limits_{k\to\infty}E\left(\eta(k)^T\eta(k)\right)=0\Leftrightarrow\lim\limits_{k\to\infty}E\left(\eta(k)^TH\eta(k)\right)=0.
\end{equation}
where $K$ is positive definite symmetric, $\lambda_{min}(H)$ and $\lambda_{max}(H)$ are the minimum and maximum eigenvalues of $H$ respectively.

Let $L_1(k)=E\left(\eta(k)^TH\eta(k)\right), L_2(k)=E\left(\theta(k)^T\theta(k)\right)$, which describe the tracking error and the estimate error respectively in mean square sense. We have the following lemmas on $L_1(k)$ and $L_2(k)$.
\begin{lemma}\label{lem3}
	Under Assuptions \ref{ass1} and \ref{ass2}, we can get that the tracking error $L_1(k)$ satisfies 
	\begin{align}
	L_1(k)\le& \left(1-\frac{1}{2\lambda_Hk}\right)L_1(k-1)+\frac{2\lambda_H\lambda_\phi d_*}{k}L_2(k-1)\notag\\&+O(f(k-1))+O\left(\frac{1}{k^2}\right),  k>T_1,
	\label{iter:L1}
	\end{align}
	where $\lambda_H$ is the maxmum eigenvalue of the matrice $H$, $\lambda_\phi$ is the maximum eigenvalue of the matrix $(I-J_{n+1})^T\phi^TH\phi(I-J_{n+1})$,  $d_*$ is the maximum degree of the nodes in the system network,  and $T_1$ is a constant.
\end{lemma}

\begin{lemma}\label{lem2}
	Under Assuptions \ref{ass1} and \ref{ass2}, the estimate error  $L_2(k)$ satisfies
	\begin{align}
		L_2(k)\le& \left(1-\frac{2\beta f_{B}-\frac{ \lambda_{W}\lambda_L}{\alpha}-2d_*}{k}\right) L_2(k-1)\notag\\
		&+\frac{\alpha h}{k}L_1(k-1)+O\left(f(k-1)\right)+O\left(\frac{1}{k^2}\right),
		\label{iter:L2}
	\end{align}
	as $k>T_2$, where  $f_B =f(B + \mathcal{W}), B = \max_{j\in N_i, i = 1, \ldots, n}|\mathcal{B}_{ij} |$, $ \mathcal{W}$ is the bound of the projection of estimation, $\alpha>0$,  $\lambda _{W}=\lambda_{\max}(WW^T)$, $\lambda_L=\lambda_{\max}(\mathcal{L}\mathcal{L}^T)$, and $h=\frac{1}{\lambda_{\min}(H)}$, and  $d_*$ is the maximum degree of the nodes in the system network,  $T_2$ is a constant.
\end{lemma} 
	The proofs of Lemma \ref{lem3} and  Lemma \ref{lem2} are given in Appendix \ref{ap:le3} and \ref{ap:le2}, respectively.
	
	From Lemma \ref{lem3} and  Lemma \ref{lem2}, we can see that  the traking error $L_1(k)$ and the estimate error $L_2(k)$ affect each other.
	So we analyze $L_1(k)$ and $L_2(k)$ together to give the convergence analysis. Before that, we give a lemma as follows.

\begin{lemma}[Theorem 1.2.23, \cite{ref23}]\label{lem4}
	If ${x_n}$ satisfies the iterative equation 
	\begin{equation*}
		x_{k+1}=(1-a_k)x_k+b_k,\; k\ge 0,
	\end{equation*}
	where $a_k \in [0,1)$, $\sum_{k=1  }^{\infty}b_k$ converges, then
	\begin{equation*}
		x_k \to 0,\; \forall \; x_0\neq 0 \Leftrightarrow\sum_{k=1}^{\infty}a_k=\infty.
	\end{equation*}
\end{lemma}

\begin{theorem}[Convergence]\label{the1}
	For the algorithm of CRS with Assumptions \ref{ass1}-\ref{ass:cov}, if the coefficient $\beta$ in estimate satisfies
	\begin{equation*}
		\beta>\frac{l_1}{f_B},
	\end{equation*}
    where $l_1=\frac{h\lambda_{W}\lambda_L}{4\lambda_H\lambda_\phi d^*}+4\lambda_\phi^2\lambda_H^3d^2_*+d^*$, $\lambda_H, \lambda_\phi, h, \lambda_W, \lambda_L, d^*$ are the same as those in Lemma \ref{lem3} and Lemma \ref{lem2},
    we can get 
	\begin{equation} 
		E\left(\hat{x}_{ij}(k)-x_j(k)\right)^2 \to 0 
	\end{equation}
	for $j\in \mathcal{N}_i,i=1,2, \ldots n+1 ,i\neq j$ and
	\begin{equation} 
		E\left(x_i(k)-x_{n+1}(k)\right)^2 \to 0 
		\label{eq2}
	\end{equation}
	for $ \forall \; i = 1, . . . , n+1$.
	
\end{theorem}

\begin{proof}
	Considering \eqref{iter:L1} and \eqref{iter:L2}, we can get
	\begin{equation}
		\begin{cases}
			L_1(k)\le& \left(1-\frac{1}{2\lambda_Hk}\right)L_1(k-1)+\frac{2\lambda_H\lambda_\phi d_*}{k}L_2(k-1)\\
			&+O(f(k-1))+O\left(\frac{1}{k^2}\right), \\
			L_2(k)\le& \left(1-\frac{2\beta f_B-\frac{ \lambda_{W}\lambda_L}{\alpha}-2d_*}{k}\right) L_2(k-1)\\
			&+\frac{\alpha h}{k}L_1(k-1)+O(f(k-1))+O\left(\frac{1}{k^2}\right).
		\end{cases}
		\label{com:L1L2}
	\end{equation}
	Let $a=\frac{1}{2\lambda_H},b=-2\lambda_H\lambda_\phi d_*,c=2\beta f_B-\frac{\lambda_{W}\lambda_L}{\alpha}-2d_*,d=-\alpha h$.
	and $Q=\begin{pmatrix}a & b\\ d & c \end{pmatrix}$. By \eqref{com:L1L2}, we can get 
	\begin{align*}
		\left\Vert P(k) \right\Vert \le &\left\Vert \left(I-\frac{Q}{k}\right)P(k-1)+O(f(k-1))+O\left(\frac{1}{k^2}\right) \right\Vert\\
		\le &\left\Vert I-\frac{Q}{k}\right\Vert\left\Vert P(k-1) \right\Vert
		+O(f(k-1))+O\left(\frac{1}{k^2}\right),
	\end{align*}
	where $P(k)=(L_1(k), L_2(k))^T$.
	
	Let $\alpha=\frac{2\lambda_H\lambda_\phi d_*}{h}$, then $b=d$. The matrix $Q$ is symmetrical. It follows that
	$$\left\Vert I-\frac{Q}{k}\right\Vert\leq 1-\frac{\lambda_{\min}(Q)}{k}, \;\text{if}\; k>\lambda_{\max}(Q).$$
	Then, we have
	\begin{align}
		\left\Vert P(k) \right\Vert \le&\left(1-\frac{\lambda_{\min}(Q)}{k}\right)\left\Vert P(k-1) \right\Vert+O(f(k-1))\notag\\
		&+O\left(\frac{1}{k^2}\right).
		\label{eq:Pk}
	\end{align}
If  $\beta>\frac{l_1}{f_B} $, where $l_1=\frac{h\lambda_{W}\lambda_L}{4\lambda_H\lambda_\phi d^*}+4\lambda_\phi^2\lambda_H^3d^2_*+d_*$, then 	
\begin{equation*}
ac=\frac{1}{2\lambda_H}(2\beta f_B-\lambda_W\lambda_L/\alpha-2d^*)>4\lambda_H^2\lambda_\phi ^2d^2_*=b^2.
\end{equation*}
Hence, the minimum eigenvalue of the matrix $Q$ satisfies
\begin{equation*} 
\lambda _{\min }(Q)=\frac{a+c-\sqrt{(a+c)^2-4(ac-b^2)}}{2}>0.
\end{equation*}
	Duo to $\sum_{k=1}^{\infty}\frac{1}{k^2}<\infty$,	we have by Lemma \ref{lem:f} and Lemma \ref{lem4}
	$$\left\Vert P(k)\right\Vert\rightarrow 0.$$
	Hence, 
	$$L_1(k)\rightarrow0, \quad L_2(k)\rightarrow 0,$$
	which implies the theorem.
	
\end{proof}

\begin{theorem}[Convergence rate]\label{the2}
   Let the chang rate of the reference state be as $ f(k) = \frac{1}{k^{1+\varepsilon}}\; (0<\varepsilon <1) $. For the algorithm of CRS with Assumptions \ref{ass1}-\ref{ass:cov}, we can get
	$$E\left(\hat{x}_{ij}(k)-x_j(k)\right)^2=		\begin{cases}
	O\left(\frac{1}{k^{\lambda_{\min}{(Q)}}}\right),\;&\text{if}\; \varepsilon>\lambda_{\min}{(Q)};\\
	O\left(\frac{\log k}{k^{\varepsilon}}\right),\;&\text{if}\; \varepsilon=\lambda_{\min}{(Q)};\\
	O\left(\frac{1}{k^\varepsilon}\right),\;&\text{if}\; \varepsilon<\lambda_{\min}{(Q)},\\
	\end{cases} $$ and 
	$$E\left(x_i(k)-x_{n+1}(k)\right)^2=	\begin{cases}
	O\left(\frac{1}{k^{\lambda_{\min}{(Q)}}}\right),\;&\text{if}\; \varepsilon>\lambda_{\min}{(Q)};\\
	O\left(\frac{\log k}{k^{\varepsilon}}\right),\;&\text{if}\;  \varepsilon=\lambda_{\min}{(Q)};\\
	O\left(\frac{1}{k^\varepsilon}\right),\;&\text{if}\; \varepsilon<\lambda_{\min}{(Q)},\\
	\end{cases}$$ 
	for $j\in N_i, i=1, \ldots,n$, where
	\begin{equation*} 
	Q=\begin{pmatrix}\frac{1}{2\lambda_H} & -2\lambda_H\lambda_\phi d_*\\ -2\lambda_H\lambda_\phi d_* & 2\beta f_B-\frac{h\lambda_{W}\lambda_L}{2\lambda_H\lambda_\phi d_*}-2d_* \end{pmatrix}, 
	\end{equation*}
	$\lambda_H, \lambda_\phi , h, f_B,\lambda_W, \lambda_L, d^*$ are the same as those in Lemmas \ref{lem3} and \ref{lem2}.
	
\end{theorem}
	\begin{proof}
	Let $\varrho=\lambda_{\min}{(Q)}$, we have by \eqref{eq:Pk}
		\begin{align*}
		\left\Vert P(k)\right\Vert &\leq \left(1-\frac{\varrho}{k}\right)\left\Vert P(k-1)\right\Vert+O\left(\frac{1}{(k-1)^{1+\varepsilon}}\right)\\
		&=\sum_{i=1}^{k}\prod_{j=i+1}^{k}\left(1-\frac{\varrho}{j}\right) \frac{1}{i^{1+\varepsilon}}+\prod_{i=1}^{k}\left(1-\frac{\varrho}{j}\right)P(0)\\
		&=\sum_{i=1}^{k}\frac{i^\varrho}{k^\varrho}\frac{1}{i^{1+\varepsilon}}+O\left(\frac{1}{k^{\varrho}}\right)\\
		&=\frac{1}{k^\varrho}\sum_{i=1}^{k} i^{\varrho-1-\varepsilon}+O\left(\frac{1}{k^{\varrho}}\right)\\
		&=
		\begin{cases}
		O\left(\frac{1}{k^\varrho}\right),\;&\text{if}\; \varepsilon>\varrho;\\
		O\left(\frac{\log k}{k^\varepsilon}\right),\;&\text{if}\; \varepsilon=\varrho;\\
		O\left(\frac{1}{k^\varepsilon}\right),\;&\text{if}\; \varepsilon<\varrho.
		\end{cases}
	\end{align*}
   Since $E\left(\hat{x}_{ij}(k)-x_j(k)\right)^2\leq L_1(k)\leq \|P(k)\| $ and $E\left(x_i(k)-x_{n+1}(k)\right)^2\leq L_2(k)\leq \|P(k)\| $ hold for $j\in N_i, i=1, \ldots,n$, we can get the theorem.
	\end{proof}
\begin{remark}
	Theorem \ref{the2} shows that  the convergence rate mainly depends on the properties of topologies and the change rate of reference state. 
\end{remark}
\begin{remark}
	Theorems \ref{the1} and \ref{the2} hold for the case of a constant reference state, which is a special case of asymptotically convergent reference state. The results are in consistent with those in \cite{refCAC} with a constant reference state. 
\end{remark}

\section{The reference state is bounded}\label{sec:bod}
In this section, the state of leader is assumed to be bounded, which can be described as the following assumption. 
\begin{assumption}\label{ass:f}
	There exists a constant $\varsigma<\infty$ such that 
	$$\|x_{n+1}(k)\|<\varsigma, \;\forall k=1, 2, \ldots.$$
\end{assumption}
By Assumption \ref{ass:f}, we can get the following proposition. 
\begin{proposition}\label{pro:f}
	There exists a constant $ \epsilon$ such that the change  rate of leader's state satisfies
	$$\|f(k)\|\leq  \epsilon.$$
\end{proposition}

 The algorithm in Section \ref{sec:cov} cannot ensure the convergence of tracking error since Lemma \ref{lem:f} does not hold. We give the following algorithm with constant step size of estimate and control gain, which can ensure the boundedness of tracking error. 
The algorithm of  bounded reference state (BRS) is given as five steps as that in the algorithm of Section \ref{sec:cov}, where the estimation and control are given as follows.
\begin{enumerate}
	\item Estimation: Each agent $i$ estimates its neighbors' states based on the RPA with a constant gain:
	\begin{align} 
		\hat{x}_{ij}(k)=
		&\Pi _\mathcal{W}\Big\{ \hat{x}_{ij}(k-1 )\vphantom{\frac{\beta }{k}} \notag\\ 
		&+\beta \left(\mathcal{F}\left(\mathcal{B}_{ij}-\hat{x}_{ij}(k-1)\right)-\mathcal{S}_{ij}(k)\right)\Big\},
		\label{eq:esbN}
	\end{align}
	where $i=1, \ldots, n$, $j\in N_i$, $\beta$ is a constant.
	\item Control: Based on the estimates, each agent designs the control with a constant:
	\begin{equation}
		u_i(k)=-\frac{1}{N}\sum_{j=1  }^{n+1}a_{ij}\left[x_i(k)-\hat{x}_{ij}(k)\right],\; i=1,2, \ldots ,n, 
		\label{eq:ub}
	\end{equation}
	where $N$ is a constant.
\end{enumerate}

By the controller \eqref{eq:ub}, the state of agent is updated as follows
\begin{align*} 
x_i(k+1)&=x_i(k)-\frac{1}{N}\sum_{j=1  }^{n+1}a_{ij}\left[x_i(k)-\hat{x}_{ij}(k)\right],
\end{align*}
which can be written in a vector form 
\begin{equation} 
x(k+1)=\left(I-\frac{ \mathcal{L}}{N}\right)x(k)+\frac{ M}{N}\theta(k)+\Gamma(k),
\label{up:Ngain}
\end{equation}
where $x(k)=(x_1(k), x_2(k), \ldots, x_{n=1}(k))^T$, $\mathcal{L}$ is the $(n+1)\times (n+1)$ Laplacian Matrix of the network $\mathcal{G}$, $M $ is defined in \eqref{def:M}, $\theta(k)$ is defined in \eqref{thetak}, and $\Gamma(k)=\left[0, 0, \ldots, 0, f(k)\right]^T$. 
Since $\eta(k)=\phi(I-J_{n+1})x(k)$, we have
\begin{align}
&\eta(k+1)\notag\\
=&\left(I-\frac{\tilde{\mathcal{L}}}{N}\right)\eta(k)+\frac{\phi(I-J_{n+1})M}{N}\theta(k)\notag\\
&+\phi\left(I-J_{n+1}\right)\Gamma(k),
\label{update:eta}
\end{align}



By algorithm of BRS with estimate \eqref{eq:esbN} and control \eqref{eq:ub}, we can get the properties of tracking error $L_1(k)=E\left(\eta(k)^TH\eta(k)\right)$ and estimate error $L_2(k)=E\left(\theta(k)^T\theta(k)\right)$ as follows.

\begin{lemma}\label{lem:BLCG1}
		Under Assumptions \ref{ass1}, \ref{ass2} and \ref{ass:f},  the tracking error $L_1(k)$ by  algorithm of BRS satisfies
	\begin{align}
	L_1(k)\le& 3\left(1-\frac{1}{N\lambda_H}\right)L_1(k-1)+\frac{3\lambda_\phi  d_*}{N^2}L_2(k-1)\notag\\
	&+3\lambda_\phi  \epsilon^2, \quad N>2\lambda_H\lambda_L,
	\label{L3}
	\end{align}
		where $ \lambda_H$, $\lambda_L$, $\lambda_\phi $  and $d_*$ are the same as those in Lemmas \ref{lem3} and \ref{lem2}, $ \epsilon$ is defined in Proposition \ref{pro:f}.
\end{lemma}

\begin{lemma}\label{lem:BLCG2}
	Under Assumptions \ref{ass1}, \ref{ass2} and \ref{ass:f}, the mean square error of state estimation $L_2(k)$ by algorithm of BRS follows
\begin{align}
L_2(k)\leq &3\left(1-\frac{\beta f_B}{2}\right)\lambda_ML_2(k-1)\notag\\
&+\frac{3\left(1-\beta f_B/2\right)\lambda_{W}\lambda_L h}{N^2}L_1(k-1)\notag\\
&+3\left(1-\frac{\beta f_B}{2}\right)\lambda_{W} \epsilon^2+\frac{nd_*\beta^2}{4}
\label{L4}
\end{align}
	where $f_B, \lambda_L, \lambda_W, \lambda_L, h, d_*$ are the same as those in Lemma \ref{lem2}, $\lambda_M=\lambda_{\max}\left(\left(I-\frac{WM}{N}\right)^T\left(I-\frac{WM}{N}\right)\right)$, $ \epsilon$ is defined  in Proposition \ref{pro:f}.
\end{lemma}
The proof of Lemmas \ref{lem:BLCG1} and \ref{lem:BLCG2} are given in Appendix \ref{ap:BLCG1} and \ref{ap:BLCG2}, respectively.

\begin{theorem}\label{the3}
   For algorithm of BRS with Assumptions \ref{ass1}, \ref{ass2} and \ref{ass:f}, if the  coefficient $\beta$ in estimate \eqref{eq:esbN} satisfies
   \begin{equation}
   \left(1-\frac{\beta f_B}{2}\right)^2<\frac{1-a^2-b^2}{c^2+d^2-(ad-bc)^2},
   \label{con:beta}
   \end{equation}
   where $a=3\left(1-\frac{1}{N\lambda_H}\right)$, $b=\frac{3\lambda_\phi d_*}{N^2}$, $c=\frac{3\lambda_{W}\lambda_Lh}{N^2}$, $d=3\lambda_M$, $N>2\lambda_H\lambda_L$, $\lambda_H$, $\lambda_\phi$, $d_*$, $\lambda_{W}$,  $\lambda_L$, $h$, $\lambda_M$ are the same as those in Lemmas \ref{lem:BLCG1} and \ref{lem:BLCG2}, we can get that for any $ \delta>0$ there exitsts a $K>0$ such that 
%
	\begin{equation} 
		E\left(\hat{x}_{ij}(k)-x_j(k)\right)^2 \le \frac{\left\Vert D \right\Vert}{1-\left\Vert Q \right\Vert}+ \delta, \; k>K
		\label{eq1}
	\end{equation}
	for $j\in \mathcal{N}_i,i=1,2, \ldots n+1,$ and
	\begin{equation} 
		E\left(x_i(k)-x_{n+1}(k)\right)^2 \le \frac{\left\Vert D \right\Vert}{1-\left\Vert Q \right\Vert}+ \delta, \; k>K
	\end{equation}
	for $ \forall \quad i = 1, . . . , n+1$,
	where 
	$$Q=\begin{pmatrix}a, & b\\ \left(1-\frac{\beta f_B}{2}\right)c, & \left(1-\frac{\beta f_B}{2}\right)d \end{pmatrix},$$
	and $$ D=\begin{pmatrix}3\lambda_\phi  \epsilon^2\\ 3\left(1-\frac{\beta f_B}{2}\right)\lambda_{W} \epsilon^2+\frac{nd_*\beta^2}{4} \end{pmatrix}.$$

\end{theorem}

\begin{proof}[Proof of Theorem~{\rm\ref{the3}}]
	Considering the tracking error \eqref{L3} and the estimation \eqref{L4} together, we can get 
\begin{equation}
	\begin{cases}
	L_1(k)\le& 3\left(1-\frac{1}{N\lambda_H}\right)L_1(k-1)+\frac{3\lambda_\phi  d_*}{N^2}L_2(k-1)\notag\\&+3\lambda_\phi  \epsilon^2 \\
	L_2(k) \le & 3\left(1-\frac{\beta f_B}{2}\right)\lambda_ML_2(k-1) \\
	&+\frac{3\left(1-\beta f_B/2\right)\lambda_{W}\lambda_L h}{N^2}L_1(k-1)\\
	&+3\left(1-\frac{\beta f_B}{2}\right)\lambda_{W} \epsilon^2+\frac{nd_*\beta^2}{4}
	\end{cases}
\end{equation}
Let $Z(k)=\begin{pmatrix}L_1(k) \\ L_2(k) \end{pmatrix} $. Then, 
\begin{align*}
\left\Vert Z(k) \right\Vert &\le \left\Vert QZ(k-1)+D \right\Vert\\
&\le\left\Vert Q \right\Vert  \left\Vert Z(k-1) \right\Vert+\left\Vert D \right\Vert\\
&\le\left\Vert Q\right\Vert^{k}  \left\Vert Z(0) \right\Vert+\left\Vert D \right\Vert\sum_{i=0}^{k-1}\left\Vert Q \right\Vert\\
&\le\left\Vert Q\right\Vert^{k}  \left\Vert Z(0) \right\Vert+\frac{\left\Vert D \right\Vert \left[1-\left\Vert Q \right\Vert^{k}\right]}{1-\left\Vert Q \right\Vert}
\end{align*}
If $\beta$ satisfies \eqref{con:beta}, we can get that
$$\left\Vert Q\right\Vert=\sqrt{\lambda_{\max}(Q^TQ)}<1. $$
Sequently, 
$$ \lim_{k\rightarrow\infty}\left\Vert Q\right\Vert^{k}  \left\Vert Z(0) \right\Vert+\frac{\left\Vert D \right\Vert \left[1-\left\Vert Q \right\Vert^{k}\right]}{1-\left\Vert Q \right\Vert}=\frac{\left\Vert D \right\Vert}{1-\left\Vert Q \right\Vert}.$$
That is to say, for any $\delta>0$ there exitsts a $K>0$ such that
 $$\left\Vert Q\right\Vert^{k}  \left\Vert Z(0) \right\Vert+\frac{\left\Vert D \right\Vert \left[1-\left\Vert Q \right\Vert^{k}\right]}{1-\left\Vert Q \right\Vert}\leq \frac{\left\Vert D \right\Vert}{1-\left\Vert Q \right\Vert}+ \delta, k>K.$$
 Hence, we can get the theorem.
%
\end{proof}

\begin{remark}
	Theorem \ref{the3} shows that the tracking error and estimate error by algorithm of BRS are bounded, and the bound depends on the change rate of reference state if the topology is given. The bigger the change rate of reference state, the greater the error bound of tracking.
\end{remark}
\begin{remark}
	By the proof of Theorem \ref{the3}, we can see that the tracking error  and the estimate error will be in the bound  with an exponential rate.
\end{remark}

\section{Numerical Simulation}\label{sec:sim}

Consider a multi-agent system with $1$ leader and $4$ followers, where the leader is agent $5$. The network topology is shown in the following figure.
\begin{figure}[!htp]
	\centering
	\includegraphics[width=1.8in]{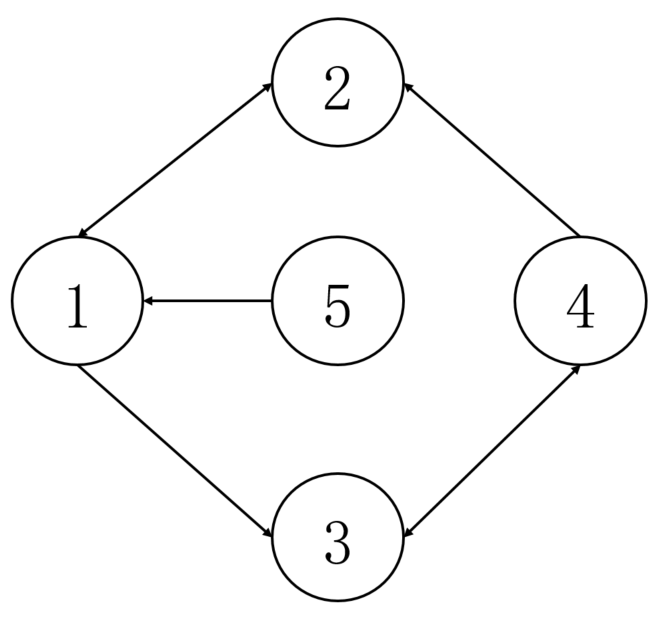}
	\caption{Network topology of three multi-agents.}
\end{figure}

The neighbor set of the agent are as follows:
\noindent  $\mathcal{N}_1=\{2,5\}, \mathcal{N}_2=\{1,4\}, \mathcal{N}_3=\{1,4\}, \mathcal{N}_4=\{3\}$, where $\mathcal{N}_i$ is the set of agent $i$'s neighbors. As the leader, agent 5 cannot receive information from the followers.
The Laplacian matrix  $\mathcal{L}$ is  $\mathcal{L}=\begin{bmatrix}2&-1&0&0&-1\\ -1&2&0&-1&0\\ -1&0&1&0&0 \\0&0&-1&1&0\\ 0&0&0&0&0 \end{bmatrix}$,
$\Theta=\begin{bmatrix}1&0&1.802&0.45&-1.25\\ 1&0&-2.25&0.8&-0.56\\ 1&-1&-2.25&0.8&-0.56 \\1&1&1&1&1 \\1&0&0&0&0 \end{bmatrix}$,
and $\Theta^{-1}=\begin{bmatrix}0&0&0&0&1\\ 0&1&-1&0&0\\ 0.194&-0.156&-0.086&0.11&0.06 \\ 0.242&0.134&0.301&0.54&-1.22 \\ -0.436&-0.98&0.785&0.35&0.28\end{bmatrix}$.
Joldan standard type of Laplacian matrix is $\tilde{\mathcal{L}}=\begin{bmatrix}
2&0&0&0\\0&3.247&0&0 \\ 0&0&0.2&0 \\0&0&0&1.55 \end{bmatrix}$.
Then, we can find $H=\begin{bmatrix}
0.25&0&0&0\\0&0.154&0&0 \\ 0&0&2.5&0 \\0&0&0&0.322 \end{bmatrix}$ such that $H\tilde{\mathcal{L}}+\tilde{\mathcal{L}}^TH=I$. Then we can get $h=\lambda_{min}^{-1}(H)=5, c_1=3.8, \lambda_{\Theta}=\lambda_{max}\left(\Theta^T\Theta\right)=2.85, \lambda_n=\lambda_5=3.247, \lambda_{c_2}=\lambda_{max}(I-J_{n+1})^T\phi^TH\phi(I-J_{n+1})=3.5$,  the
parameter $l_1$ and $l_2$ in conditions of Theorem 2 can
be calculated by 
$l_1=\frac{h^2\lambda_{W}\lambda_L}{4\lambda_\phi }+\frac{4c_1^2d_*^2}{\lambda_\phi ^3}+d_*=86.54$ and $l_2=\frac{8\lambda_\phi ^2d_*^2}{c_1^3-2\lambda_\phi ^2}+\frac{c_1h\lambda_{W}\lambda_L}{2\lambda_\phi }+2d_*+1=25.69$.

Let the initial value of the agents $x(0) = [x_1(0), x_2(0), x_3(0),x_4(0),x_5(0)]^T =[-30,-10,20,10,15]^T$, the initial estimation of each agent for the neighbor satisfies $[\hat{x}_{1}(0),\hat{x}_{2}(0), \hat{x}_{3}(0), \hat{x}_{4}(0),\hat{x}_{4}(0)]^T = [-4,-5,2,0,5]^T$, the given
threshold, and the state bound are  $ \mathcal{B}_{ij}=0,i,j=1,2,3,4,5$ and $ \mathcal{W} = 50$, respectively. The distribution of the noise is $N(0, 10)$. Then $f_B = f\left(\left|\mathcal{B}\right| + \mathcal{W}\right) = 0.293$. 

\begin{figure}[!thp]
	\centering
	\subfloat[The states]{\includegraphics[width=1.72in]{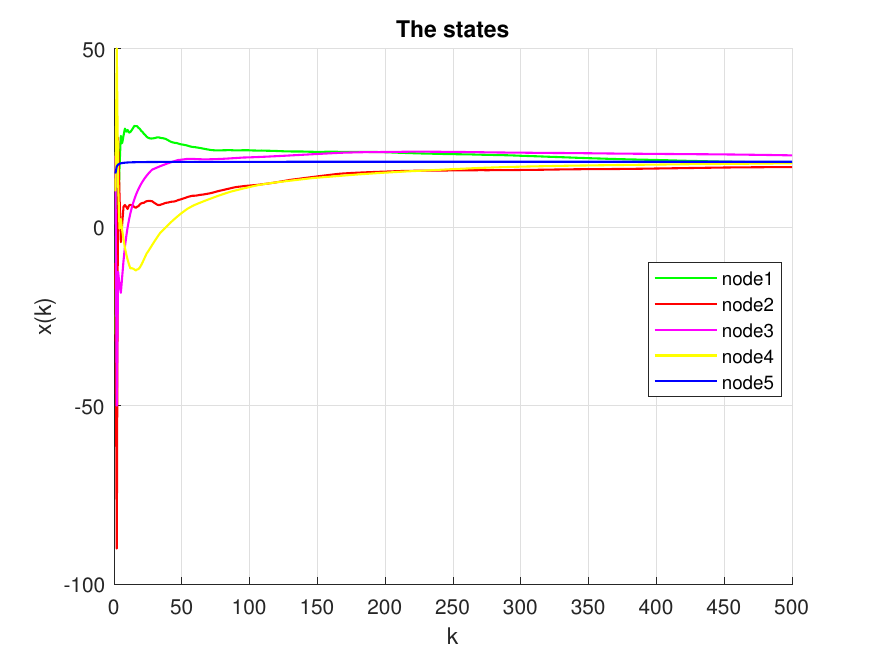}%
		\label{fig:statePlike}}
	\subfloat[The estimated value]{\includegraphics[width=1.72in]{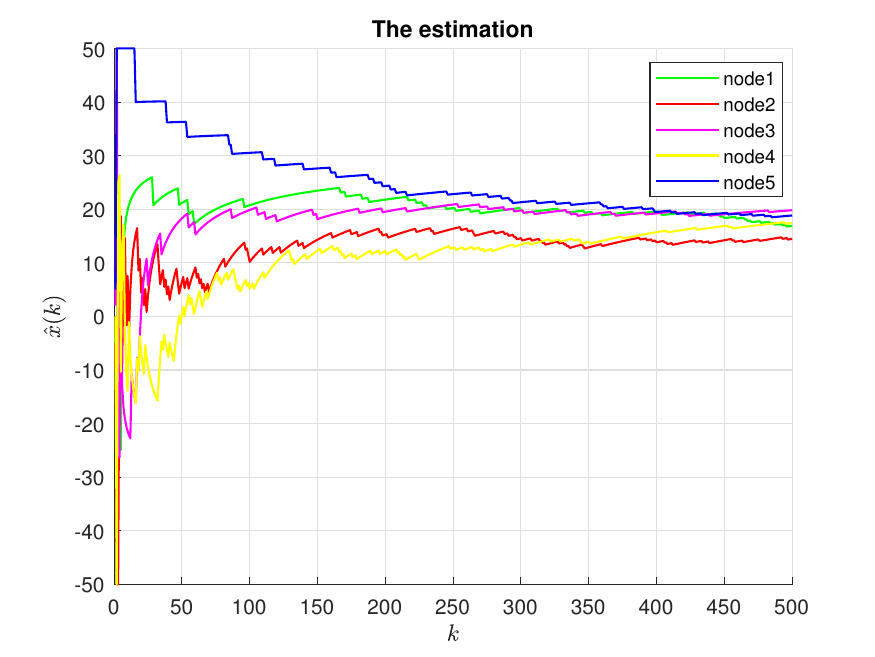}%
		\label{fig:estimatePlike}}
	\hfil
	\subfloat[The estimate error]{\includegraphics[width=1.72in]{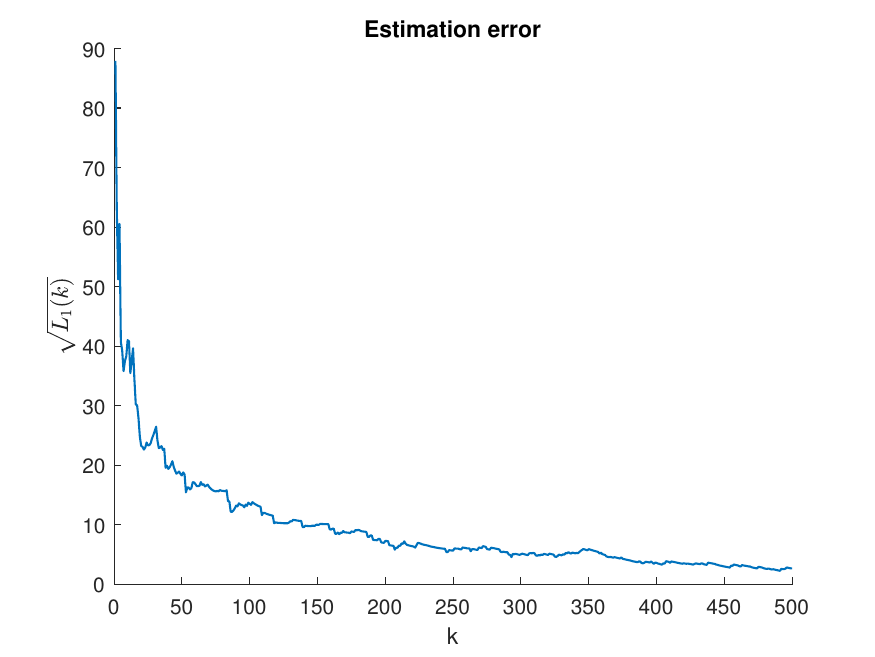}%
		\label{fig:estimateerror}}
	\subfloat[The tracking error]{\includegraphics[width=1.72in]{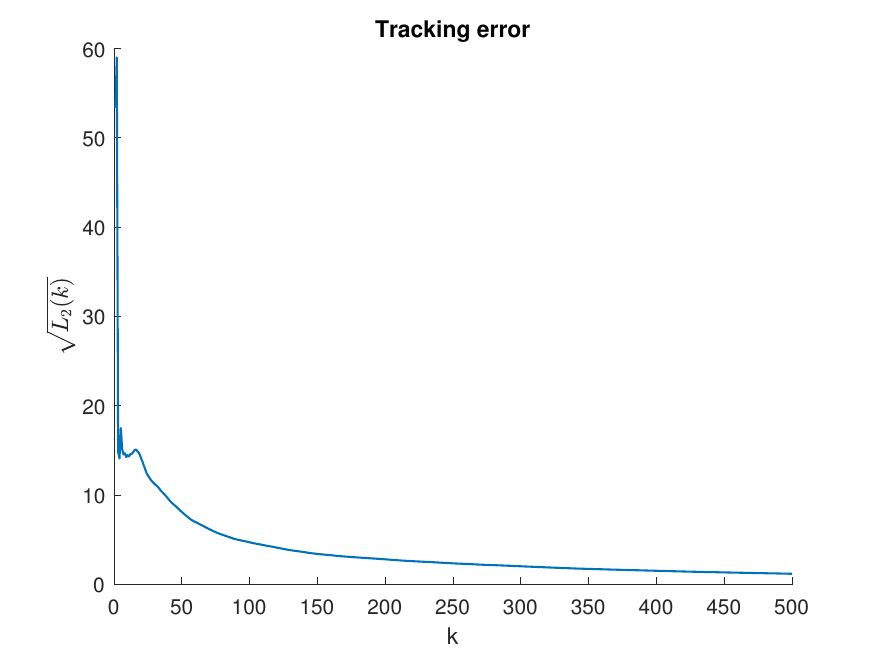}%
		\label{fig:trackerror}}
	\caption{Algorithm of CRS with $\beta =150$.}
	\label{fig_sim1}
\end{figure}
 \begin{figure}[!thp]
	\centering
	\subfloat[The states]{\includegraphics[width=1.72in]{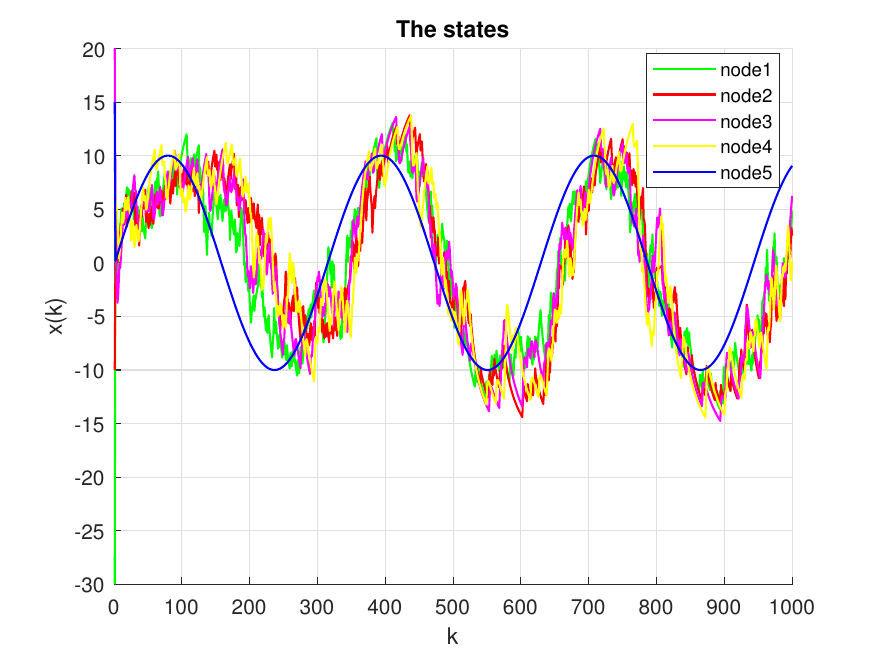}%
		\label{fig:statePlikecons}}
	\subfloat[The estimated value]{\includegraphics[width=1.72in]{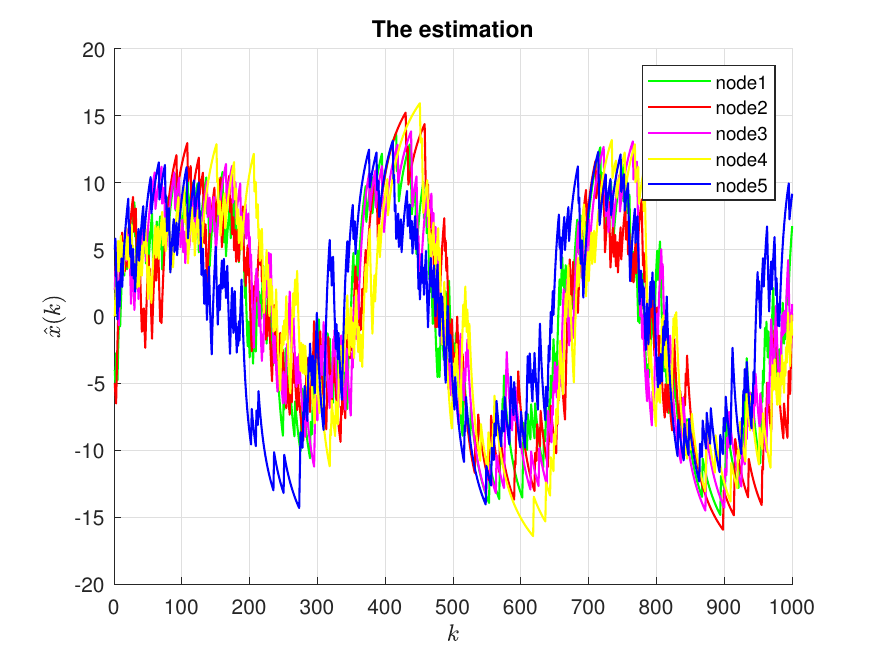}%
		\label{fig:estimatePlikescons}}
	\hfil
	\subfloat[The estimate error]{\includegraphics[width=1.72in]{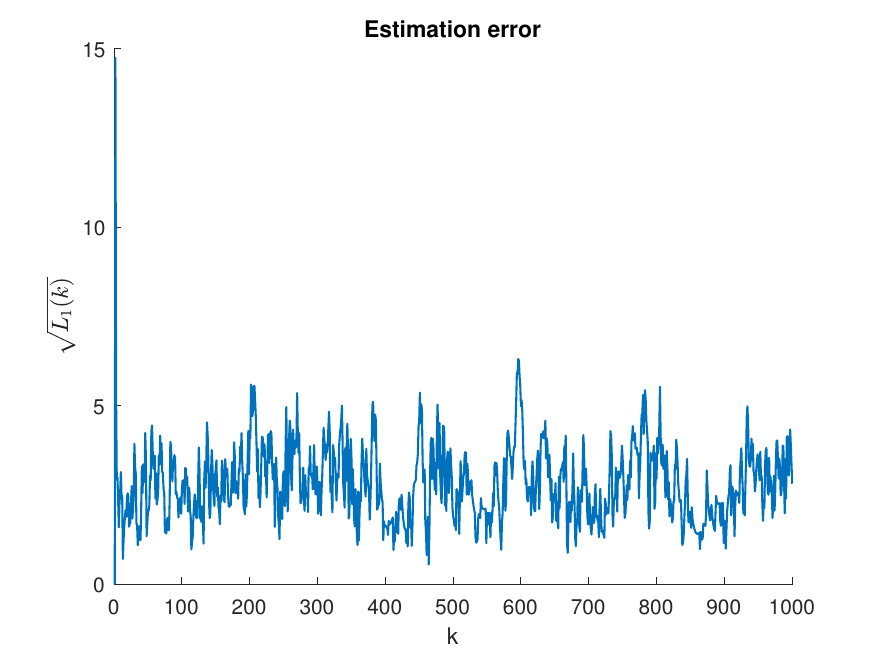}%
		\label{fig:estimateerrorPlikescons}}
	\subfloat[The tracking error]{\includegraphics[width=1.72in]{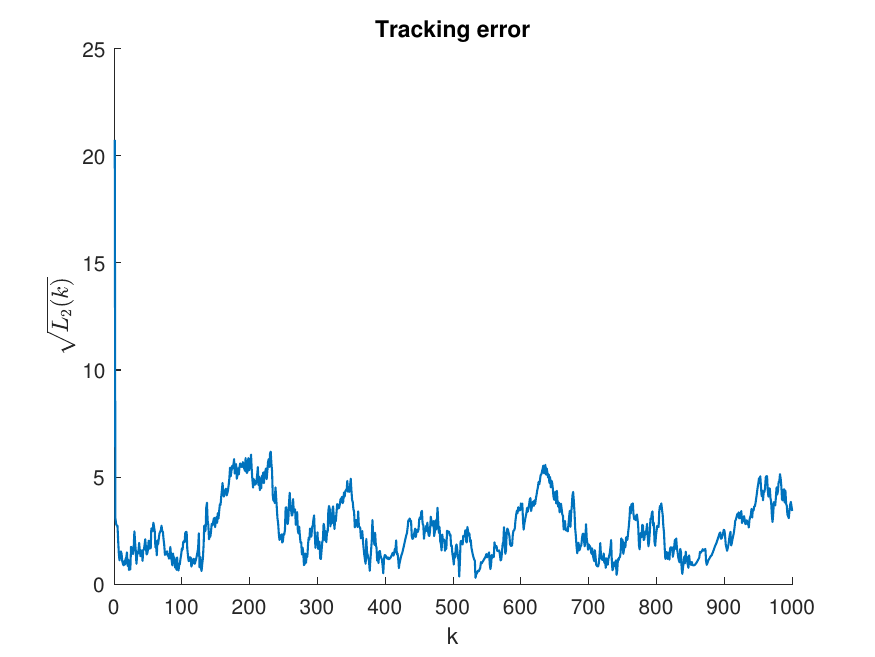}%
		\label{fig:trackerrorPlikescons}}
	\caption{ Algorithm of BRS with $\beta =3$.}
	\label{fig_sim3}
\end{figure}

\subsection{The reference state is asymptotically convergent}
Consider the reference state is asymptotically convergent with state updating as follows 
$$x_5(k+1)=x_5(k)+f(k),$$
where $f(k)=\frac{1}{k^2}$, which satisfies $\sum_{k=1}^{\infty}f(k)<\infty$. Choosing $\beta=150$ in the algorithm of CRS,  we can get the trajectories of the states and estimates of agents in Figure (\ref{fig:statePlike}) and (\ref{fig:estimatePlike}), respectively. The estimate error and the tracking error are given in (\ref{fig:estimateerror}) and (\ref{fig:trackerror}). From Figure \ref{fig_sim1}, we can see that the followers can converge to the state  of leader, which is consistent with Theorem \ref{the1}.

\subsection{The reference state is bounded}
 
 Consider the leader $x_5(k)=10sin(0.02k)$. By  algorithm of BRS, the trajectories of states, estimates, estimate and tracking errors are given in  Figure \ref{fig_sim3}, which shows the followers can track the leader in a bounded area. That is consistent with Theorem \ref{the3}.

%
%
%

\section{Conclusions}\label{sec:con}
This paper consider  consensus tracking problems with a time-varying reference state and binary-valued communication. For the asymptotically convergent reference state, the algorithm with decaying step sizes and control gains is given. It is proved that the followers can converge to the leader and the convergence rate depends on the change rate of the  reference state. For the bounded  reference state, an algorithm with constant step sizes and control gains is given. It is proved that the followers can achieve to a bounded neighborhood of the leader with an exponential rate. In the future, we can consider some complex tasks, such as formation control, obstacle avoidance and path planning with  limited communication and switching topology.


%

\appendices

\section{Proof of Lemma  \ref{lem3}}\label{ap:le3}
By the state updating \eqref{update3}, the tracking error $L_1(k)$ satisfies
\begin{align*}
	&L_1(k)\notag\\
	=&E \left[\Bigg( \left(I-\frac{\tilde{\mathcal{L}}}{k}\right)\eta(k-1)+\frac{\phi(I-J_{n+1}) M}{k}\theta(k-1)\right.\notag\\
	&\quad+\phi(I-J_{n+1})\Gamma(k-1)\Bigg)^TH \notag\\ 
	&\quad \Bigg( \left(I-\frac{\tilde{\mathcal{L}}}{k}\right)\eta(k-1)+\frac{\phi(I-J_{n+1}) M}{k}\theta(k-1)\notag\\
	&\left.\quad+\phi(I-J_{n+1})\Gamma(k-1)\Bigg)\right].\\
	\end{align*}
It can be written as
\begin{align}
	&L_1(k)\notag\\
	=&E\bigg(\eta(k-1)^T\left(I-\frac{\tilde{\mathcal{L}}}{k}\right)^TH\left(I-\frac{\tilde{\mathcal{L}}}{k}\right)\eta(k-1)\bigg)\notag\\
	&
	+\frac{1}{k^2}E\left(\theta(k-1)^TM^TAM\theta(k-1)\right)\notag\\
	&+E\left(\Gamma(k-1)^TA\Gamma(k-1)\right)\notag\\
	&+\frac{2}{k}E\Bigg(\theta(k-1)^TM^T(I-J_{n+1})\phi^TH\left(I-\frac{\tilde{\mathcal{L}}}{k}\right)\eta(k-1)\Bigg)\notag\\
	&+\frac{2}{k}E\left(\theta(k-1)^TM^TA\Gamma(k-1)\right)\notag\\
	&+2E\Bigg(\Gamma(k-1)^T(I-J_{n+1})\phi^TH\left(I-\frac{\tilde{\mathcal{L}}}{k}\right)\eta(k-1)\Bigg),
	\label{eq:L1}
\end{align}
where matric $A=(I-J_{n+1})^T\phi^TH\phi(I-J_{n+1})$, vector $\Gamma(k)=\left[0, 0, \ldots, 0, f(k)\right]^T$.

Using the properties of positive definite matrices, we can get 
\begin{align}
	&E\left(\eta(k-1)^T\left(I-\frac{\tilde{L}}{k}\right)^TH\left(I-\frac{\tilde{\mathcal{L}}}{k}\right)\eta(k-1)\right)\notag\\
	&=E\left(\eta(k-1)^T\left(H-\frac{\left(H\tilde{\mathcal{L}}+\tilde{\mathcal{L}}^TH\right)}{k}+\frac{\tilde{\mathcal{L}}^TH\tilde{\mathcal{L}}}{k^2}\right)\right.\notag\\
	&\quad \quad\left.\eta(k-1)\right)\notag\\
	&=E\left(\eta(k-1)^T\left(H-\frac{ I}{k}+\frac{\tilde{\mathcal{L}}^TH\tilde{\mathcal{L}}}{k^2}\right)\eta(k-1)\right)\notag\\
	&=E\left(\eta(k-1)^T\left(H-\frac{ I}{2k}\right)\eta(k-1)\right),
	\label{eq:neta}
\end{align}
if $k>2\lambda_{max}(\tilde{\mathcal{L}}^TH\tilde{\mathcal{L}})$.
Together with $$E\left(\eta(k-1)^T\eta(k-1)\right)\geq \frac{E\left(\eta(k-1)^TH\eta(k-1)\right)}{\lambda_{max}(H)},$$ we can get that
\begin{align}
&E\left(\eta(k-1)^T\left(I-\frac{\tilde{L}}{k}\right)^TH\left(I-\frac{\tilde{\mathcal{L}}}{k}\right)\eta(k-1)\right)\notag\\
&\le \left(1-\frac{1}{2\lambda_{max}(H)k}\right)E\left(\eta(k-1)^TH\eta(k-1)\right)\notag\\
&=\left(1-\frac{c_1}{k}\right)L_1(k-1),\;c_1=\frac{1}{2\lambda_{max}(H)}.
\label{pr:L1}
\end{align}
Besides, we have
\begin{align*}
	 \frac{1}{k^2}E\left(\theta(k-1)^TM^TAM\theta(k-1)\right)\le
	 O\left(\frac{1}{k^2}\right),
\end{align*}
and
\begin{align*}
	&E\left(\Gamma(k-1)^TA\Gamma(k-1)\right) \le \lambda_\phi f^2(k)=\lambda_\phi f^2(k-1),
\end{align*}
where $\lambda_\phi =\lambda_\phi =\lambda_{max}((I-J_{n+1})^T\phi^TH\phi(I-J_{n+1}))$.
By holder inequality, we have 
\begin{align*}
	&E\Bigg(\theta(k-1)^TM^T(I-J_{n+1})\phi^TH\left(I-\frac{\tilde{\mathcal{L}}}{k}\right)\eta(k-1)\Bigg)\\
	\leq & \sqrt{\alpha E\left(\eta(k-1)^T\left(I-\frac{\tilde{L}}{k}\right)^TH\left(I-\frac{\tilde{\mathcal{L}}}{k}\right)\eta(k-1)\right)}\\
	&\sqrt{\frac{1}{\alpha}E\left(\theta(k-1)^TM^TAM\theta(k-1)\right)}\\
	\leq& \frac{1}{2}\left( \alpha L_1(k)+\frac{\lambda_\phi d_*}{\alpha}L_2(k)\right),
\end{align*}
where $\alpha>0$. Let $\alpha=c_1/2$, we have
\begin{align*} 
&\frac{2}{k}E\Bigg(\theta(k-1)^TM^T(I-J_{n+1})\phi^TH\left(I-\frac{\tilde{\mathcal{L}}}{k}\right)\eta(k-1)\Bigg)\\
\leq &\frac{c_1}{2k}L_1(k)+\frac{2\lambda_\phi d_*}{c_1k}L_2(k).
\end{align*}
Since the estimates and the states of agents are bounded by Remark \ref{rem:estimate} and Proposition \ref{rem:pro1}, it follows that
$$ \frac{2}{k}E\left(\theta(k-1)^TM^TA\Gamma(k-1)\right)=O\left(\frac{1}{k}f(k-1)\right),$$
and 
$$E\Bigg(\Gamma(k-1)^T(I-J_{n+1})\phi^TH\left(I-\frac{\tilde{\mathcal{L}}}{k}\right)\eta(k-1)\Bigg)=O(f(k-1)). $$
By inequaltiy \eqref{eq:L1}, we have that
\begin{align*}
	L_1(k)\le& \left(1-\frac{c_1}{2k}\right)L_2(k-1)+\frac{2\lambda_\phi d_*}{c_1k}L_2(k-1)\notag\\&+O(f(k-1))+O\left(\frac{1}{k^2}\right),\; k>T_1, 
\end{align*}
where $T_1=2\lambda_{max}(\tilde{\mathcal{L}}^TH\tilde{\mathcal{L}})$.

\section{Proof of Lemma \ref{lem2}}\label{ap:le2}
\vspace*{12pt}

According to the definition of the estimated error, we can get
\begin{align*}
	L_2(k)=E\left(\theta(k)^T\theta(k)\right) = \sum_{i=1}^{n}\sum_{j \in N_i}E\left(\hat{x}_{ij}(k)-x_j(k)\right)^2.
\end{align*}

By the estimation \eqref{up:estimate}, the state updating \eqref{agentupdate1}, we can get
\begin{align}
	&E\theta_{ij}^2(k) = E\left(\hat{x}_{ij}(k)-x_j(k)\right)^2 \notag\\
	&\leq E \left(\hat{x}_{ij}(k-1)+\frac{\beta}{k}(\mathcal{F}\left(\mathcal{B}_{ij}-\hat{x}_{ij}(k-1)\right)-S_{ij}(k))\right.\notag\\
	&\quad \left.-x_j(k) \right)^2 \notag\\
	&= E\Bigg( \hat{x}_{ij}(k-1)-x_j(k-1)+\frac{\beta}{k}\Big(\mathcal{F}\left(\mathcal{B}_{ij}-\hat{x}_{ij}(k-1)\right)\notag\\&\quad -S_{ij}(k)\Big)+\frac{1}{k}\sum_{p=1}^{n+1}a_{jp}  (x_j(k-1)-\hat{x}_{jp}(k-1)) \Bigg)^2 \notag\\
	&= E\Bigg( \theta_{ij}(k-1)+\frac{\beta}{k}\big(\mathcal{F}\left(\mathcal{B}_{ij}-\hat{x}_{ij}(k-1)\right)-S_{ij}(k)\big) \notag\\&\quad\quad\quad +\frac{1}{k}\sum_{p=1}^{n+1}a_{jp}(x_j(k-1)-\hat{x}_{jp}(k-1)) \Bigg)^2, \notag\\
	&\quad\quad\text{if}\; j=1, \ldots, n,
	\label{up:theta}
\end{align}
and 
\begin{align}
	E\theta_{ij}^2(k)\leq &E\Bigg( \theta_{ij}(k-1)+\frac{\beta}{k}\big(\mathcal{F}\left(\mathcal{B}_{ij}-\hat{x}_{ij}(k-1)\right)-S_{ij}(k)\big) \notag\\&\quad\quad\quad -f(k-1) \Bigg)^2, \quad\text{if}\; j= n+1.
	\label{up:theta2}
\end{align}
Let \begin{align*}
	\alpha(k)= &\theta(k)+\frac{\beta}{k+1}(\hat{F}(k)-S(k+1))\\
 	&+W\left( \frac{1}{k+1}\mathcal{L}x(k)-\frac{1}{k+1}M\theta(k)-\Gamma(k) \right),
\end{align*}
 where $W$ is defined in \eqref{def:W}.

 By \eqref{up:theta} and \eqref{up:theta2}, we can obtain
\begin{align}
	L_2(k) &\le E\left(\alpha(k-1)^T\alpha(k-1)\right)\notag\\
	&=E\left(\theta^T(k-1)\theta(k-1) \right)\notag\\
	&\quad+\frac{2\beta}{k}E\left(\theta^T(k-1)\left(\hat{F}(k-1)-S(k)\right)\right)\notag\\
	&\quad+\frac{2\theta^T(k-1) W}{k}E\left(\mathcal{L}x(k-1)-M\theta(k-1) \right)\notag\\
	&\quad+O\left(\frac{1}{k^2}\right)+O\left(f(k) \right)+O\left(\frac{f(k)}{k}\right).
	\label{up:L2}
\end{align}
By mean value principle, we have
\begin{align*}
&E\big(\mathcal{F}\left(\mathcal{B}_{ij}-\hat{x}_{ij}(k-1)\right)-S_{ij}(k)\big)\\
=&E\big(\mathcal{F}\left(\mathcal{B}_{ij}-\hat{x}_{ij}(k-1)\right)-\mathcal{F}\left(\mathcal{B}_{ij}-x_j(k)\right)\big)\\
=&-f(\zeta_{ij}(k))(\hat{x}_{ij}(k-1)-x_j(k)),
\end{align*}
where $\zeta_{ij}(k)$ is the middle of $\hat{x}_{ij}(k-1)$ and $x_j(k)$, which is bounded. Hence, we can get that
\begin{align*} 
&\frac{2\beta }{k}E\left[\theta (k-1)^T(\hat{F}(k-1)-S(k))\right] \\
 &\leq -\frac{2\beta f_B}{k}L_2(k-1)+O\left(\frac{f(k)}{k}\right)+O\left(\frac{1}{k^2}\right).  
 \end{align*}
where $\beta$ is a step size for estimation, where  $f_B =f(B + \mathcal{W}), B = \max_{j\in N_i, i = 1, \ldots, n}|\mathcal{B}_{ij} |$, $ \mathcal{W}$ is the bound of the projection of estimation. 

Since 
\begin{align*} 
	&\frac{2}{k}E\left[\theta (k-1)^TW\mathcal{L}x(k-1)\right]\\
	 \leq &\frac{2}{k}\sqrt{E[\theta (k-1)^TW{\mathcal{L}}\mathcal{L}^TW^T\theta (k-1)]} \\ 
	&\quad\sqrt{E[x(k-1)^T x(k-1)]} \\ 
	\leq &\frac{2}{k}\sqrt{E[\theta (k-1)^TW{\mathcal{L}}\mathcal{L}^TW^T\theta (k-1)]} \\
	&\quad \sqrt{E[\xi(k-1)^T(I-J_{n+1})^{-1}(I-J_{n+1})^{-1}\xi(k-1)]} \\
	 \leq &\frac{2}{k}\sqrt{\frac{\lambda _{W}\lambda_L}{\alpha }L_2(k-1)\alpha \lambda _{J}E(\eta^T(k-1)\eta(k-1)) } \\ 
	 	 \quad\leq &\frac{2}{k}\sqrt{\frac{\lambda _{W}\lambda_L}{\alpha }L_2(k-1)   \frac{\alpha\lambda _{J}}{\lambda_{\min}(H)}L_1(k-1) } \\
	  \leq &\frac{1}{k}\left(\frac{\lambda _{W}\lambda_L}{\alpha }L_2(k-1)+\alpha h L_1(k-1)\right)
\end{align*}
where $\lambda _{W}=\lambda_{\max}(WW^T)$, $\lambda_L=\lambda_{\max}(\mathcal{L}\mathcal{L}^T)$, $\lambda _{J}=\lambda_{\max}((I-J_{n+1})^{-1}(I-J_{n+1})^{-1})=1 $, and $h=\frac{1}{\lambda_{\min}(H)}$. 
Also,
\begin{align*}
	&-\frac{2}{k}E\left[\theta (k-1)^T WM\theta (k-1)\right]  \leq \frac{2d_*}{k}L_2(k-1)
\end{align*}
where $d_*$ is the maximum degree of the nodes in the system network. By \eqref{up:L2}, we can get the lemma.

\section{Proof of Lemma \ref{lem:BLCG1}}\label{ap:BLCG1}
\vspace*{12pt}

By \eqref{update:eta} and the definition of $L_3(k)$, it follows
\begin{align*}
	&L_3(k)=E\left(\eta(k)^TH\eta(k)\right)\\
	&E \Bigg( \left(I-\frac{\tilde{\mathcal{L}}}{N}\right)\eta(k-1)+\frac{\phi(I-J_{n+1}) M}{N}\theta(k-1)\\
	&+\phi(I-J_{n+1})\Gamma(k-1)\Bigg)^TH \\ 
	&\Bigg( \left(I-\frac{\tilde{\mathcal{L}}}{N}\right)\eta(k-1)+\frac{\phi(I-J_{n+1}) M}{N}\theta(k-1)\\
	&+\phi(I-J_{n+1})\Gamma(k-1)\Bigg)\\
	&\leq 3E\bigg(\eta(k-1)^T\left(I-\frac{\tilde{\mathcal{L}}}{N}\right)^T\phi^TH\phi\left(I-\frac{\tilde{\mathcal{L}}}{N}\right)\\
	&\quad\quad\eta(k-1)\bigg)
	+\frac{3}{N^2}E\left(\theta(k-1)^TM^TAM\theta(k-1)\right)\\
	&+3E\left(\Gamma(k-1)^TA\Gamma(k-1)\right)
\end{align*}
where $A=(I-J_{n+1})^T\phi^TH\phi(I-J_{n+1})$, $\Gamma(k)=\left[0, 0, \ldots, 0, f(k)\right]^T$.

Similar to \eqref{eq:neta} and \eqref{pr:L1} in Appendix \ref{ap:le3},  we can get 
\begin{align*}
	&E\left(\eta(k-1)^T\left(I-\frac{\tilde{L}}{N}\right)^TH\left(I-\frac{\tilde{\mathcal{L}}}{N}\right)\eta(k-1)\right)\\
	&\leq \left(1-\frac{1}{2N\lambda_H}\right)L_1(k-1),\; as \; N>2\lambda_H\lambda_L,
\end{align*}
where $\lambda_H=\lambda_{\max}(H)$,  $\lambda_L=\lambda_{\max}(\mathcal{L}\mathcal{L}^T)$.
Together with
\begin{align*}
	\frac{1}{N^2}E\left(\theta(k-1)^TM^TAM\theta(k-1)\right)\le
	\frac{\lambda_\phi d_*}{N^2}L_2(k-1),
\end{align*}
and
\begin{align*}
	&E\left(\Gamma(k-1)^TA\Gamma(k-1)\right) \le \lambda_\phi  \epsilon^2,
\end{align*}
where $\lambda_\phi=\lambda_{max}(A)$, we can get the lemma.

\section{Proof of Lemma \ref{lem:BLCG2}}\label{ap:BLCG2}
By \eqref{eq:esbN}, we have
\begin{align}
&E\theta_{ij}^2(k) = E\left(\hat{x}_{ij}(k)-x_j(k)\right)^2 \notag\\
&\leq E \left(\hat{x}_{ij}(k-1)+\beta(\mathcal{F}\left(\mathcal{B}_{ij}-\hat{x}_{ij}(k-1)\right)-S_{ij}(k))
-x_j(k) \right)^2 \notag\\
&= E(\hat{x}_{ij}(k-1)-x_j(k))^2\notag\\
&\quad+2\beta E\Big[(\mathcal{F}\left(\mathcal{B}_{ij}-\hat{x}_{ij}(k-1)\right)-S_{ij}(k))\notag\\
&\quad \quad\quad\quad\quad(\hat{x}_{ij}(k-1)-x_j(k)) \Big]\notag\\
&\quad +\beta^2 E(\mathcal{F}\left(\mathcal{B}_{ij}-\hat{x}_{ij}(k-1)\right)-S_{ij}(k))^2
\label{up:tha}
\end{align}
Since 
\begin{align*}
&E\Big[(\mathcal{F}\left(\mathcal{B}_{ij}-\hat{x}_{ij}(k-1)\right)-S_{ij}(k))(\hat{x}_{ij}(k-1)-x_j(k)) \Big]\\
=&E\Big[\big(\mathcal{F}\left(\mathcal{B}_{ij}-\hat{x}_{ij}(k-1)\right)-\mathcal{F}\left(\mathcal{B}_{ij}-x_j(k)\right)\big)\\
&\quad (\hat{x}_{ij}(k-1)-x_j(k)) \Big]\\
=&-E\Big[f(\zeta_{ij}(k))(\hat{x}_{ij}(k-1)-x_j(k))^2\Big],
\end{align*}
where $\zeta_{ij}(k)$ is the middle of $\hat{x}_{ij}(k-1)$ and $x_j(k)$, 
and 
\begin{align*}
	&E(\mathcal{F}\left(\mathcal{B}_{ij}-\hat{x}_{ij}(k-1)\right)-S_{ij}(k))^2\\
	=&E\left(\mathcal{F}\left(\mathcal{B}_{ij}-\hat{x}_{ij}(k-1)\right)-\mathcal{F}\left(\mathcal{B}_{ij}-x_j(k)\right)\right.\\
	&\left.\quad+\mathcal{F}\left(\mathcal{B}_{ij}-x_j(k)\right)-S_{ij}(k)\right)^2\\	
	=&E\left(\mathcal{F}\left(\mathcal{B}_{ij}-\hat{x}_{ij}(k-1)\right)-\mathcal{F}\left(\mathcal{B}_{ij}-x_j(k)\right)\right)^2\\
	&+E\left(\mathcal{F}\left(\mathcal{B}_{ij}-x_j(k)\right)-S_{ij}(k)\right)^2\\
	=&E\Big[f^2(\zeta_{ij}(k))(\hat{x}_{ij}(k-1)-x_j(k))^2\Big]\\
	&+E\left(\mathcal{F}\left(\mathcal{B}_{ij}-x_j(k)\right)(1-\mathcal{F}\left(\mathcal{B}_{ij}-x_j(k)\right))\right)\\
	\leq &E\Big[f^2(\zeta_{ij}(k))(\hat{x}_{ij}(k-1)-x_j(k))^2\Big]+\frac{1}{4},
\end{align*}
it follows that
\begin{align*}
E\theta_{ij}^2(k)\leq &E(1-\beta f(\zeta_{ij}(k)))^2(\hat{x}_{ij}(k-1)-x_j(k))^2+\frac{\beta^2}{4}\\
\leq &\left(1-\frac{\beta f_B}{2}\right)E(\hat{x}_{ij}(k-1)-x_j(k))^2+\frac{\beta^2}{4},
\end{align*}
if $\beta\leq \frac{1}{2f(0)}$, where  $f_B =f(B + \mathcal{W}), B = \max_{j\in N_i, i = 1, \ldots, n}|\mathcal{B}_{ij} |$, $ \mathcal{W}$ is the bound of the projection of estimation. 
According to the definition of the estimated error, we can get
\begin{align}
&L_2(k)\notag\\
&=\sum_{i=1}^{n}\sum_{j \in N_i}E\left(\hat{x}_{ij}(k)-x_j(k)\right)^2\notag\\
&\leq \left(1-\frac{\beta f_B}{2}\right)\sum_{i=1}^{n}\sum_{j \in N_i}E(\hat{x}_{ij}(k-1)-x_j(k))^2+\frac{nd_*\beta^2}{4}.
\label{eq:L22}
\end{align}

Let \begin{align*}
\gamma(k)= \theta(k)+W\left( \frac{1}{N}\mathcal{L}x(k)-\frac{1}{N}M\theta(k)-\Gamma(k) \right),
\end{align*}
we have
\begin{align*}
	&E(\gamma^T(k)\gamma(k))\\
	=&E\left(\theta(k)+W\left( \frac{1}{N}\mathcal{L}x(k)-\frac{M}{N}\theta(k)-\Gamma(k) \right) \right)^T\\
	&\quad \left(\theta(k)+W\left( \frac{1}{N}\mathcal{L}x(k)-\frac{M}{N}\theta(k)-\Gamma(k) \right) \right)\\
	\leq& 3E\left(\theta(k)^T\left(I-\frac{WM}{N}\right)^T\left(I-\frac{WM}{N}\right)\theta(k)\right)\\
	&+\frac{3}{N^2}E\left(x^T(k)\mathcal{L}^TW^TW\mathcal{L}x(k)\right)\\
	&+3E\left(\Gamma(k)^TW^TW\Gamma(k)\right)\\
	\leq& 3\lambda_ML_2(k)+\frac{3\lambda_{W}\lambda_L h}{N^2}L_1(k)+3\lambda_{W} \epsilon^2,
\end{align*}
where $\lambda_M=\lambda_{\max}\left(\left(I-\frac{WM}{N}\right)^T\left(I-\frac{WM}{N}\right)\right)$,  $\lambda _{W}=\lambda_{\max}(WW^T)$, $\lambda_L=\lambda_{\max}(\mathcal{L}\mathcal{L}^T)$, and $h=\frac{1}{\lambda_{\min}(H)}$. 
Since
\begin{align*}
&\sum_{i=1}^{n}\sum_{j \in N_i}E(\hat{x}_{ij}(k-1)-x_j(k))^2\\
=&E(\gamma^T(k-1)\gamma(k-1))\\
\leq &3\lambda_ML_2(k-1)+\frac{3\lambda_{W}\lambda_L h}{N^2}L_1(k-1)+3\lambda_{W} \epsilon^2,
\end{align*}
we have by \eqref{eq:L22}
\begin{align*}
L_2(k)\leq &3\left(1-\frac{\beta f_B}{2}\right)\lambda_ML_2(k-1)\\
&+\frac{3\left(1-\frac{\beta f_B}{2}\right)\lambda_{W}\lambda_L h}{N^2}L_1(k-1)\\
&+3\left(1-\frac{\beta f_B}{2}\right)\lambda_{W} \epsilon^2+\frac{nd_*\beta^2}{4}
\end{align*}


\section{Biography Section}

\vspace{11pt}

\vspace{-33pt}
\begin{IEEEbiography}[{\includegraphics[width=1in,height=1.25in,clip,keepaspectratio]{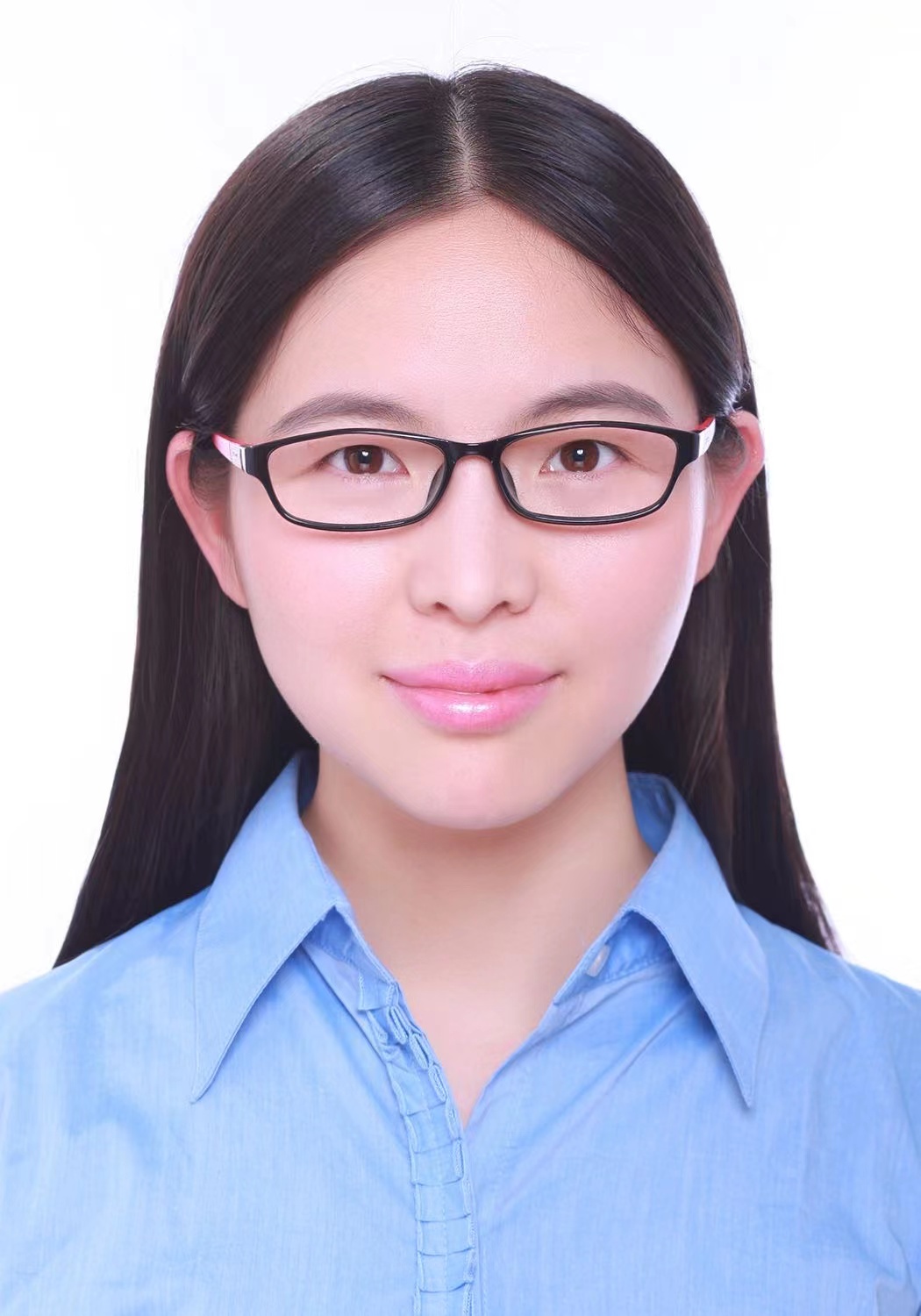}}]{Ting Wang} received the B.S. degree in mathematics and applied mathematics from Xidian University, Xi’an, China, in 2012, and the Ph.D. degree in operational research and cybernetics from the Academy of Mathematics and Systems Science, Chinese Academy of Sciences, Beijing, China, in 2017. 

She is currently an Associate Professor with the School of Intelligence Science and Technology, University of Science and Technology Beijing. Her current research interests include identification and control of quantized systems, and distributed control of multi-agent systems.
\end{IEEEbiography}
\begin{IEEEbiography}[{\includegraphics[width=1in,height=1.25in,clip,keepaspectratio]{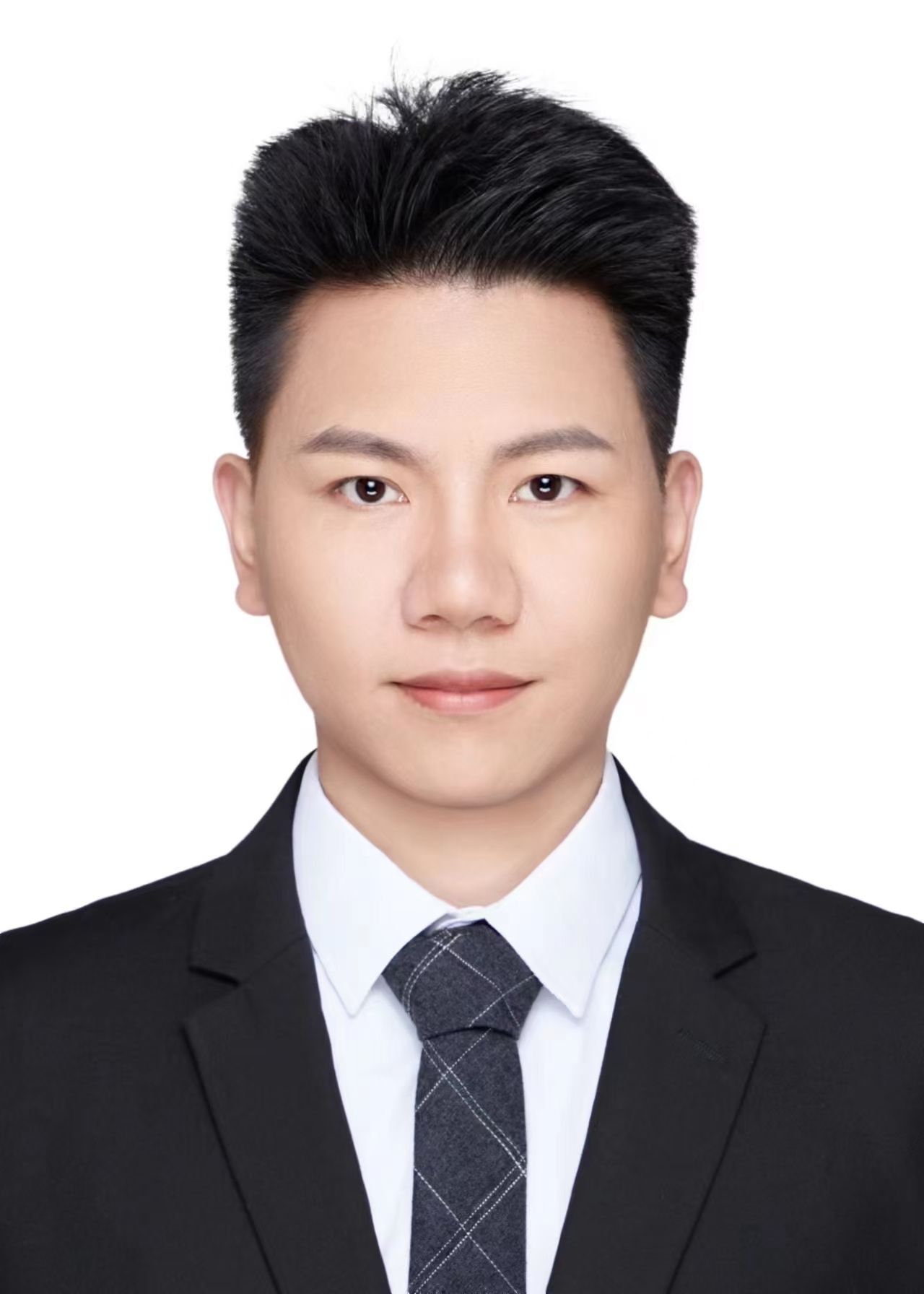}}]{Qiu Zhuangzhuang} graduated from Tianjin University of Technology majoring in automation in 2019. Now he is a postgraduate student majoring in Control Science and Engineering at University of Science and Technology Beijing. His research interest is distributed consensus tracking control of multi-agent systems.
\end{IEEEbiography}

\begin{IEEEbiography}[{\includegraphics[width=1in,height=1.25in,clip,keepaspectratio]{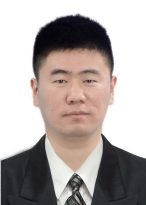}}]{Xiaodong Lu} received the Ph. D. degree from the School of Control Science and Engineering, Shandong University in 2019. From 2022 to 2024, he was a Postdoctoral Researcher with the Institute of Systems Science, Chinese Academy of Sciences, China. He is currently an Associate Professor with the School of Automation and Electrical Engineering, University of Science and Technology Beijing. His research interests include multi-agent systems and stochastic systems.
\end{IEEEbiography}

\begin{IEEEbiography}[{\includegraphics[width=1in,height=1.25in,clip,keepaspectratio]{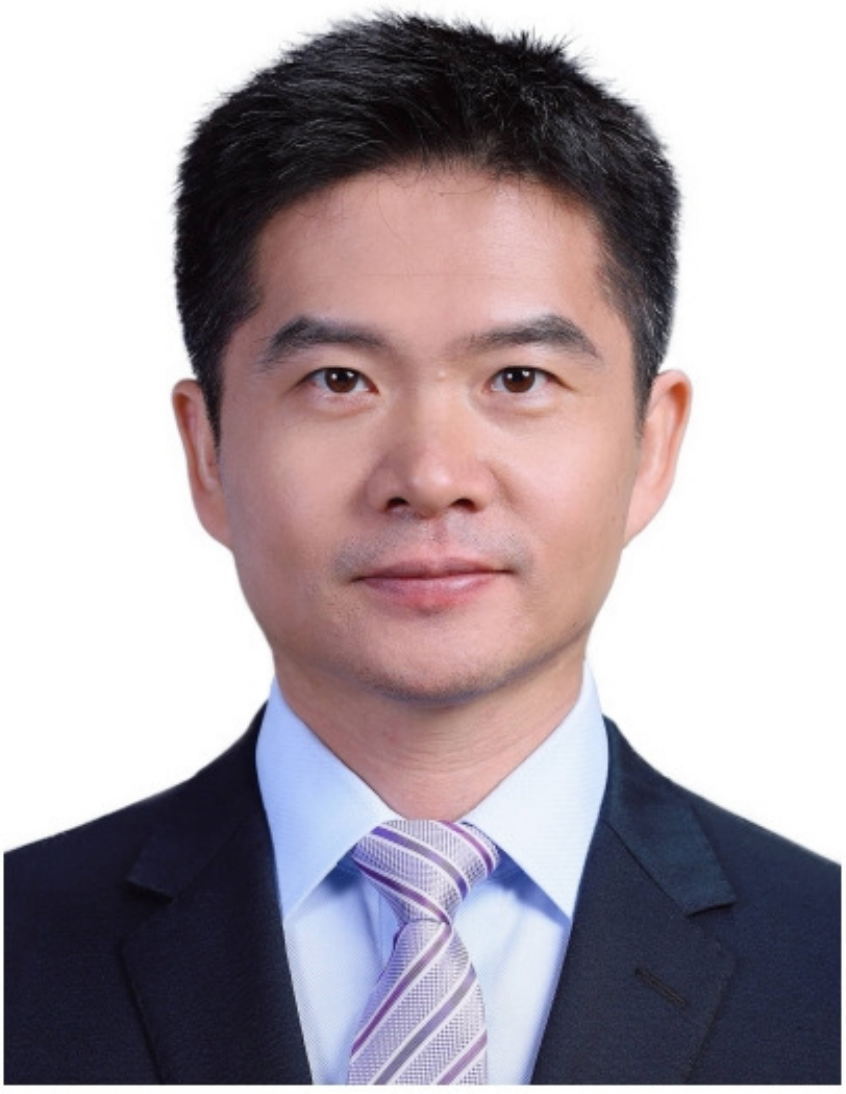}}]{Yanlong Zhao} received the B.S. degree in mathematics from Shandong University, Jinan, China, in 2002, and the Ph.D. degree in systems theory from the Academy of Mathematics and Systems Science (AMSS), Chinese Academy of Sciences (CAS), Beijing, China, in 2007. Since 2007, he has been with the AMSS, CAS, where he is currently a full Professor. His research interests include identification and control of quantized systems, information theory and modeling of financial systems. He has been a Deputy Editor-in-Chief of \emph{Journal of Systems and Science and Complexity}, an Associate Editor of \emph{Automatica}, \emph{SIAM Journal on Control and Optimization}, and \emph{IEEE Transactions on Systems, Man and Cybernetics: Systems}. He served as a Vice-President of Asian Control Association and a Vice-President of IEEE CSS Beijing Chapter, and is now a Vice General Secretary of Chinese Association of Automation (CAA) and the Chair of Technical Committee on Control Theory (TCCT), CAA.
\end{IEEEbiography}



\vfill

\end{document}